\documentclass[a4paper,fleqn]{cas-sc} 

\usepackage{amsmath,amsfonts,amssymb,amsthm,booktabs,color,graphicx,hyperref,url}
\usepackage{caption}
\usepackage[authoryear]{natbib}
\usepackage{placeins}

\usepackage{bm}
\usepackage{mathtools}
\usepackage{multirow}
\usepackage{dsfont}

\shorttitle{Testing the equality of parameters in fixed and increasing dimension}
\shortauthors{Romero-Madro\~nal et~al.}

\newcommand{\E}{\mathbb{E}}
\newcommand{\V}{\mathbb{V}}
\newcommand{\C}{\mathbb{C}\mathrm{ov}}
\newcommand{\black}{\textcolor{black}}

\newcommand{\tas}{\xrightarrow{a.s.}}
\newcommand{\Prob}{\operatorname{Pr}}
\newcommand{\R}{\mathbb{R}}
\newcommand{\tr}{\operatorname{tr}}
\newcommand{\tl}{\xrightarrow{\hspace{0.1cm}\mathcal{L}\hspace{0.1cm}}}
\newcommand{\tp}{\xrightarrow{\hspace{0.1cm}\Prob\hspace{0.1cm}}}

\begin{document}

\ExplSyntaxOn
\cs_gset:Npn \__make_tbl_caption:nn #1#2
  {
    \parbox{\l_tbl_width_dim}
    {
      \rmfamily\small
      \textbf{\color{scolor}#1:}~#2
    }
  }
\ExplSyntaxOff

\newcommand{\arxivfooter}{Journal of Multivariate Analysis (2026), doi:10.1016/j.jmva.2026.105730}
\ExplSyntaxOn
\cs_gset:Npn \__first_footerline:
  {
    \group_begin:
    \small \sffamily
    \__short_authors: :~ { \rmfamily \itshape \arxivfooter }
    \group_end:
  }
\ExplSyntaxOff

\let\printorcid\relax
\let\WriteBookmarks\relax
\def\floatpagepagefraction{1}
\def\textpagefraction{.001}

\title[mode=title]{\texorpdfstring{Testing the equality of  parameters in fixed and increasing dimension}{Testing the equality of  parameters in fixed and increasing dimension}}

\author[1,2]{M. Romero-Madroñal}

\author[1,2]{M.R. Sillero-Denamiel}
\author[1,2]{M.D. Jiménez-Gamero}
\affiliation[1]{organization={Department of Statistics and Operations Research, Universidad de Sevilla}, country={Spain}}
\affiliation[2]{organization={Instituto de Matem\'aticas de la Universidad de Sevilla}, country={Spain}}
\cortext[cor1]{Corresponding author. E-mail address: \texttt{mrmadronal@us.es}.}


\begin{abstract}
This paper proposes a general and unified framework for testing the equality of a broad class of parameters, defined as a smooth function of expectations of symmetric kernels, across multiple independent populations. We consider two test statistics, a Wald-type statistic and an ANOVA-type statistic. The asymptotic distribution of the first one is derived under a fixed-dimension regime, whereas the second one is studied under both fixed and increasing-dimension regimes, where the parameter dimension diverges with the sample size. Based on these limiting distributions, we construct test procedures enabling asymptotically exact inference without parametric assumptions. Additionally, an alternative null distribution estimator based on a weighted bootstrap approximation is studied, which is applicable to the ANOVA-type statistic  under a  fixed-dimension regime. The finite-sample performance and computational efficiency of the proposed procedures are evaluated through an extensive simulation study and a real dataset application.
\end{abstract}

\maketitle

\renewcommand{\figurename}{Fig.}

\theoremstyle{plain}
\newtheorem{theorem}{Theorem}
\newtheorem{proposition}{Proposition}
\newtheorem{lemma}{Lemma}
\newtheorem{corollary}{Corollary}

\theoremstyle{definition}
\newtheorem{definition}{Definition}
\newtheorem{remark}{Remark}
\newtheorem{example}{Example}
\newtheorem{assumption}{Assumption}

\newenvironment{Proof}{\paragraph{\textbf{\upshape Proof:}}}{\hfill$\square$ \par \vspace{3pt}}
\newenvironment{Proofx}[1]{\paragraph{\textbf{\upshape Proof of Theorem \ref{#1}:}}}{\hfill$\square$ \par \vspace{3pt}}
\newenvironment{Prooflx}[1]{\paragraph{\textbf{\upshape Proof of Lemma \ref{#1}:}}}{\hfill$\square$ \par \vspace{3pt}}
\newenvironment{Proofpx}[1]{\paragraph{\textbf{\upshape Proof of Proposition \ref{#1}:}}}{\hfill$\square$ \par \vspace{3pt}}

\section{Introduction}

Comparing a parameter across multiple populations is a long-standing problem in statistics. In this paper, assuming that independent samples are available from each population, the problem of testing the null hypothesis $ H_0: \theta_1 =\cdots= \theta_k $ is considered, where each $ \theta_i = f(\eta_i) \in \mathbb{R}^d $ ($ d \geq 1 $), $ \eta_i \in \mathbb{R}^q $ ($ q \geq 1$) is a vector of expectations of symmetric kernels of possibly different degrees, and $f$ is differentiable in a neighborhood of $ \eta_i$, $i\in \{1,\ldots,k\}$. 

Classical methods such as one-way ANOVA \citep{Fisher1925} and Bartlett's test \citep{Bartlett1937} were developed to test for the equality of means and variances, respectively, assuming normality. While foundational, these procedures are designed for specific parameters. The approach in this paper is non-parametric in the sense that no distribution is assumed for the data. In modern statistical practice, researchers are increasingly concerned with comparing a wide range of parameters across multiple populations. These parameters often arise in applied contexts: for example, the Gini index is used in economics to assess income inequality \citep{Davidson2009}, and correlation coefficients are central in psychology \citep{Paul1989,OlkinFinn1995}. In addition, other functionals such as quantiles \citep{Ditzhaus2021}, correlation matrices \citep{Jennrich1970}, and both univariate and multivariate coefficients of variation \citep{Bhoj1993,Ditzhaus2025} have also been subject to formal inference procedures. 

{\color{black}
Most of the test statistics proposed for these problems are quadratic forms in the estimators of the parameters. When the matrix of the quadratic form is an estimate of the asymptotic covariance matrix of the estimators, we have a Wald-type statistic (WTS). Examples are the test in \citet{Konietschke2015} for comparing means, for correlation matrices in \citet{Jennrich1970}, for two-sample comparisons of the Gini index, coefficients of variation and correlation coefficients in \citet{Davidson2009},  \citet{Bhoj1993}, and  \citet{OlkinFinn1995}, and for $k$-sample comparisons of multivariate coefficients of variation in \citet{Ditzhaus2025} and quantiles in \citet{Ditzhaus2021}. The unstandardized quadratic form, the ANOVA-type statistic (ATS), was introduced by \citet{Brunner1997} for means. For general  matrices,  \citet{SattlerPauly2024} give the asymptotic distribution of the corresponding quadratic form and apply it to correlation matrices.  
The ATS was also used for high-dimensional mean vectors, where the estimated covariance matrix cannot be reliably inverted \citep{Bai1996}.

Other approaches have also been considered. Empirical likelihood \citep{owen2001empirical} provides a nonparametric alternative based on a likelihood ratio, whose resulting test statistics are, to first order, quadratic forms under the null hypothesis. This methodology has been applied for comparing two \citep{Jing1995} and $k$ \citep{YuElBarmiYing2011} means, two Gini indices \citep{Lv2017}, $k$ variances \citep{Sang2021}, and two high-dimensional mean vectors \citep{WangPengQi2013}, among others. Other approaches include max-type \citep{CaiLiuXia2014} and general norms \citep{Xu2016} statistics, which have been proposed for comparing two high-dimensional mean vectors.}

A common feature of many of the parameters considered in the previous paragraphs is that their estimators can be expressed as functions of $U$-statistics, a broad class of statistics with nice asymptotic properties and wide applicability \citep{Hoeffding1948}. This observation provides a common framework for studying the problem of testing equality of parameters across populations. Based on the asymptotic normality of $U$-statistics and the Delta Method, the asymptotic normality of each $\widehat{\theta}_i$ is established. Although the exact asymptotic covariance matrices are typically unknown, consistent estimators for these matrices can be obtained for each $\theta_i$ using the jackknife method, as originally proposed by \citet{Arvesen1969}.

As discussed above, quadratic-form statistics have been widely used for testing equality of parameters in both fixed- and increasing-dimension settings. To provide a common framework encompassing these settings, we construct test statistics that are  quadratic forms in the estimated parameters. In particular, we consider the two canonical members of this family, the WTS and the ATS. The former is used for fixed $d$, whereas the latter is also considered when $d$ grows. The WTS relies on the inverse of the estimated covariance matrix and, assuming a fixed dimension, has an asymptotic $\chi^2$ distribution under the null hypothesis. In contrast, the ATS avoids this inversion and converges under the null to a non-pivotal limit given by a weighted sum of $\chi^2$ random variables, whose weights can be consistently estimated. A weighted bootstrap approach to approximate the null distribution of the ATS is also considered. When $d$ increases with the sample sizes, a suitably centered and scaled version of the ATS is shown to be asymptotically standard normal. 

The methodology developed in this paper encompasses many of the aforementioned procedures with a few exceptions, such as those based on quantiles, as they cannot be expressed as smooth functions of expectations of kernels. Furthermore, while many of the existing approaches are restricted to the comparison of two populations, the methodology developed here allows for testing equality across an arbitrary number $k \geq 2$ of populations.

{\color{black} While the classical setting assumes that $k$ and $d$ are fixed, recent literature has explored various asymptotic extensions. For instance, the case where the number of populations $k$ tends to infinity has been studied for univariate parameters estimated by $U$-statistics \citep{JimenezGamero2025,Romero-Madronal2025}. Regarding the increasing-dimension regime addressed in this work, related questions have been investigated for specific scenarios: several works address the equality of high-dimensional mean vectors for $k=2$ \citep{Huang2022}, MANOVA-type procedures for $k\geq 2$ have been developed as the dimension increases \citep{Schott2007,Srivastava2013,Hu2017,ZhangGuoZhou2017}, and numerous contributions consider equality of covariance matrices (see, e.g.,~\citealp{Li2012,Zheng2020}). A common feature of these works is that allowing $d$ to grow at rates comparable to or faster than the sample sizes typically requires imposing strong structural assumptions, such as factorial designs in covariance testing. By focusing on the regime $d/n\to 0$, our framework accommodates growing dimensions while allowing the asymptotic analysis to be carried out without imposing such restrictive structural conditions.}

The \black{remainder} of the paper is organized as follows.  Section~\ref{section:framework} introduces the statistical framework, defining the estimators based on $U$-statistics, describing their asymptotic properties, detailing their asymptotic variance estimation via the jackknife method, and formulating the hypothesis of parameter equality.  Section~\ref{section:test_stat} presents the proposed test statistics: a WTS and an ATS. The asymptotic distribution of both statistics is derived under a fixed-dimension regime. Additionally, an approximation based on a weighted bootstrap is proposed for the ATS.  Section~\ref{sub:incdim} studies the ATS under an increasing-dimension regime.   Section~\ref{numerical:experiments} describes some numerical experiments.  Section~\ref{sec:simulations} summarizes the outputs of an extensive simulation study designed to evaluate the finite-sample performance of the proposed procedures, while  Section~\ref{practical:guidelines} provides practical guidelines for implementation. Full details of the simulation study as well as an application to real data illustrating the methodology are provided in the Supplementary Material. Finally,  Section~\ref{section:conclusions} summarizes the conclusions and outlines potential lines for future work. Technical proofs are presented in  Section~\ref{section:proofs}. Some previous results on $U$-statistics are given in the Supplementary Material.

\section{Statistical framework}\label{section:framework}

\subsection{Statistical formulation}

Let $ k \in \mathbb{N}, k \geq 2 $ fixed. Let $ X_{1}, \dots, X_{k} $ be independent random elements defined on a common probability space $ (\Omega, \mathcal{F}, \Prob) $, taking values in $E$. It will be assumed that the available data is as follows:
\begin{equation}
\begin{aligned}
&\boldsymbol{X}_i = \{X_{i1}, \dots, X_{in_i} \} \text{ is a random sample of size } n_i \text{ from the distribution of } X_i, i \in \{1,\ldots,k\}, \\
&\text{and the samples } \boldsymbol{X}_1, \ldots, \boldsymbol{X}_k \text{ are independent.}
\end{aligned}
\label{eq:independent_sample}
\end{equation}
Let $ \eta_i \in \mathbb{R}^q $, with some $q\geq 1$, be defined as
\begin{equation*}
\eta_i=\E\{h(X_{i1},\ldots,X_{im})\}= \E\left[h^{(1)}(X_{i1},\ldots, X_{im_1}),\ldots, h^{(q)}(X_{i1},\ldots, X_{im_q})\right]^{\top}:=\left[\eta_i^{(1)},\ldots,\eta_i^{(q)}\right]^{\top},
\end{equation*}  
where $ m_1, \ldots, m_q \geq 1 $ and  $m=\max\{m_1,\ldots,m_q\}$,
$i \in \{1,\ldots,k\}$, with $ h^{(1)},\ldots,h^{(q)}: E^m \to \mathbb{R} $ being symmetric kernels of degree $m_1,\ldots,m_q$ respectively.  
Given the definition of the parameter $ \eta_i $,  
we estimate it using the vector of $ U $-statistics  $\widehat{\eta}_i = \left[\widehat{\eta}^{(1)}_i,\ldots,\widehat{\eta}^{(q)}_i\right]^{\top}$, where
\begin{equation*}
\widehat{\eta}^{(r)}_i=\frac{1}{n_i(m_r)}\sum_{j_1 \neq \cdots \neq j_{m_r}} h^{(r)}(X_{ij_1}, \ldots, X_{ij_{m_r}}), \quad r\in\{1,\ldots, q\},\; i \in \{1, \dots, k\},
\end{equation*}
with $ n_i(s) = n_i(n_i-1) \cdots (n_i-s+1) $. By construction, $\widehat{\eta}_i$ is an unbiased estimator of $\eta_i$ for each $i \in \{1, \dots, k\}$. 
We then define the transformed parameters and their estimators as
$\theta_i = f(\eta_i)$, 
$\widehat{\theta}_i = f(\widehat{\eta}_i)$, 
$i \in \{1, \dots, k\}$,
where $f:\mathbb{R}^q \to \mathbb{R}^d$ is a given function. Let
$\theta=\left[\theta_1^{\top},\ldots,\theta_k^{\top}\right]^{\top}$ and $\widehat{\theta}=\left[\widehat{\theta}_1^{\top},\ldots,\widehat{\theta}_k^{\top}\right]^{\top}$. 

Throughout this work, all asymptotic results are understood to hold as $n_i \to \infty$, $i \in \{1,\ldots,k\}$. Two asymptotic regimes are considered for the dimension $d$ of the parameters: the fixed-dimension regime, in which $d$ does not vary with the sample sizes, and the increasing-dimension regime, in which $d=d(n)\to\infty$ with $d/n \to 0$, where $n=n_1+\cdots+n_k$. In both regimes $k$ remains fixed.

\subsection{Basic asymptotics}

We summarize the basic asymptotic properties of the $U$-statistics, 
which allow us to derive approximate distributions for inference in the next sections. First, denote the variance components of the kernels $h^{(r)}$ as
\begin{equation*}
\zeta^{(r)}_{ic} = \V\{h^{(r)}_c(X_{i1},\ldots,X_{ic})\}, \quad  c \in \{1,\ldots, m_r\}, \; r \in \{1,\ldots, q\}, \; i \in \{1,\ldots, k\},
\end{equation*}  
where $h_c^{(r)}(x_1,\ldots,x_c)$ is defined in \black{(S1) of the Supplementary Material}  and $h_c^{(r)}(x_1,\ldots,x_c)=0$ if $c>m_r$. Likewise, define
\begin{equation*}
h_c(x_1,\ldots,x_c)=\left[h_c^{(1)}(x_1,\ldots,x_c),\ldots, h_c^{(q)}(x_1,\ldots,x_c)\right]^{\black \top}, \quad c \in \{1,\ldots, m\}.
\end{equation*}

We start by stating the Strong Law of Large Numbers (SLLN) for the $U$-statistics defined above. 
Assume that, for every $i \in \{1,\ldots,k\}$, the kernel satisfies the integrability condition  $\E\left\{\left\lvert  h(X_{i1},\ldots,X_{im})\right\rvert  \right\} < \infty$,
then the SLLN for $U$-statistics applies (see, e.g.,~Theorem A in Section 5.4 of \citet{serfling1980approximation}), i.e.,~$\widehat{\eta}_i \tas \eta_i$, 
$i \in \{1,\ldots, k\}$. Similarly, if in addition $f:\mathbb{R}^q \to \mathbb{R}^d$ is continuous at $\eta_i$, $ i \in \{1,\ldots, k\}$, it follows from the Continuous Mapping Theorem that
\begin{equation}
\widehat{\theta}_i = f(\widehat{\eta}_i) \tas f(\eta_i) = \theta_i, \quad i \in \{1,\ldots, k\}.
\label{eq:SLLNtheta}
\end{equation}

We now turn to the asymptotic distribution of the estimators. They admit a limiting normal distribution, as established in Theorem 7.1 of \citet{Hoeffding1948}, provided that the following conditions hold,
\begin{equation}
\E\left\{\rVert h(X_{i1}, \ldots, X_{im})\rVert ^2_2\right\} < \infty,  \quad
\zeta_{i1}^{(r)} > 0, \quad r \in \{1,\ldots, q\}, \; i \in \{1,\ldots, k\}.
\label{ass:finitenessh2_bounzeta}
\end{equation}
Specifically,
\begin{align}
&\sqrt{n_i}(\widehat{\eta}_i - \eta_i) \tl Z_i, \quad Z_i \sim \mathcal{N}_q(0, V_i), \quad i \in \{1,\ldots,k\},
\label{eq:limitthetai}\\
& V_i=\V\{M_h h_1(X_{i1})\}=\E\left\{M_h \widetilde{h}_1(X_{i1})\widetilde{h}_1^{\top}(X_{i1})M_h^{\top}\right\}, \quad i \in \{1,\ldots,k\},
\label{eq:definitionvi}
\end{align}
where $\widetilde{h}_1$ is defined in \black{(S2) of the Supplementary Material} and
\begin{equation}
M_h=\operatorname{diag}(m_1,\ldots,m_q).
\label{eq:Mh}
\end{equation}
A similar result can be obtained for the joint estimator $\widehat{\eta}=\left[\widehat{\eta}_1^{\top},\ldots,\widehat{\eta}_k^{\top}\right]^{\top}$ of $\eta=\left[\eta_1^{\top},\ldots,\eta_k^{\top}\right]^{\top}$. It will be assumed
that the sample sizes are comparable in the sense that
\begin{equation}
n_i/n \to \kappa_i \in (0,1), \quad i \in \{1,\ldots,k\},
\label{eq:compsamp}
\end{equation}
and thus, $\sum_{i=1}^k \kappa_i=1$. Since samples from different populations are independent, as assumed in \eqref{eq:independent_sample}, \eqref{eq:limitthetai} implies
\begin{equation}
\sqrt{n}(\widehat{\eta} - \eta) \tl Z, \quad Z  \sim \mathcal{N}_{kq}(0, V),
\label{eq:assymptheta}
\end{equation}
with $V=\operatorname{diag} \left(V_1/\kappa_1,\ldots, V_k/\kappa_k\right)$, where $\operatorname{diag}(\cdot)$ denotes a block-diagonal matrix. The previous result can be extended to each $\theta_i=f(\eta_i)$. If $f=\left[f^{(1)},\ldots,f^{(d)}\right]$ has nonzero differential at $\eta_i$, by applying the multivariate Delta Method to \eqref{eq:assymptheta} (see, e.g.,~Theorem~3.3.A of \citeauthor{serfling1980approximation}, \citeyear{serfling1980approximation}), we get that $\sqrt{n_i} ( \widehat{\theta}_i - \theta_i ) \tl Z_i$, $ Z_i\sim \mathcal{N}_d\left( 0, \Sigma_i \right)$, where
\begin{equation}
\Sigma_i=D_i \, V_i \, D_i^\top=\E\left\{D_i \,M_h \widetilde{h}_1(X_{i1})\widetilde{h}_1^{\top}(X_{i1})M_h^{\top}D_i^\top\right\},
\label{eq:definitionsigmai}
\end{equation}
and $D_i=\left(D_i^{(a,b)}\right)$,
with $D_i^{(a,b)}=f^{(a)}_b (\eta_i):=\left. (\partial/\partial x_b) f^{(a)}(x)\right\vert_{x=\eta_i}$, $ a \in \{1,\ldots,d\}$, $ b \in \{1,\ldots, q\}$.
Similarly, if \eqref{eq:compsamp} holds,
\begin{equation}
\sqrt{n} \left( \widehat{\theta} - \theta \right) \tl Z, \quad Z\sim \mathcal{N}_{kd}\left( 0, \Sigma \right), \quad \Sigma=\operatorname{diag}\left(\Sigma_1/\kappa_1, \ldots, \Sigma_k/\kappa_k\right).
\label{definition:sigma}
\end{equation}

\subsection{Variance estimation}\label{subsec:varest}

In order to make inferences, we  need to estimate the covariance matrix $ \Sigma $ in   \eqref{definition:sigma}. In some particular cases, explicit expressions for $ \Sigma $ are available. Then, the variance can be estimated by plugging consistent estimators of the parameters involved. However, in common cases $ \Sigma $ does not have an easily computable expression.

As a general procedure for estimating $ \Sigma $, we use the jackknife method. Let $ T_{i,j}$ denote the jackknife pseudovalue associated with the $j$th observation  of $\widehat{\theta}_i$, defined as
$ 
T_{i,j} := n_i \widehat{\theta}_i - (n_i-1)\widehat{\theta}_{i,(-j)},
$ 
where $\widehat{\theta}_{i,(-j)}$ denotes the estimator $\widehat{\theta}_i$ computed without the $j$th observation. 
The following lemma is an extension of Theorem 9 in \citet{Arvesen1969}, which was originally stated for scalar parameters $\theta \in \mathbb{R}$ expressible as $\theta = f(\eta)$, with $\eta$ denoting a parameter that admits a $U$-statistic estimator. 
\begin{lemma}
Suppose that \eqref{eq:independent_sample}, \eqref{ass:finitenessh2_bounzeta} and  \black{\eqref{eq:compsamp}} hold, and that each component of $f$ has continuous first partial derivatives in a neighborhood of $\eta_i$. 
Let
\begin{equation}
\widehat{\Sigma}_i=\frac{1}{n_i-1}\sum_{j=1}^{n_i}\left(T_{i,j}-\overline{T}_i\right)\left(T_{i,j}-\overline{T}_i\right)^{\top},
\label{eq:WidehatSigma}
\end{equation}
where $\overline{T}_i=(1/n_i)\sum_{j=1}^{n_i} T_{i,j}$, $i \in \{1,\ldots, k\}$. Let
$\widehat{\Sigma}=\operatorname{diag}\left(\widehat{\Sigma}_1/\widehat{\kappa}_1, \ldots, \widehat{\Sigma}_k/\widehat{\kappa}_k\right)$, with $\widehat{\kappa}_i=n_i/n$.  Then, $\widehat{\Sigma} \tp \Sigma$.

\label{Lemma:varconsist}
\end{lemma}

\subsection{Hypothesis}\label{sub:hypothesis}
We are interested in testing the equality of the parameters $\theta_1,\ldots,\theta_k$ across all populations, that is
\begin{equation*}
H_0: \theta_1 = \cdots = \theta_k \; \text{ vs } \; H_1: \exists \; i \neq j \, \text{ such that } \, \theta_i \neq \theta_j.
\end{equation*}

Let $I_d$ be the $d \times d$ identity matrix, $\mathds{1}_k$ the $k$-dimensional vector of ones, and $P_k = I_k - (1/k) \mathds{1}_k \mathds{1}_k^\top$. With this notation, the null hypothesis $H_0$ can be equivalently written as $H_0: C\theta = 0$, where $C=P_k \otimes I_d \in \mathbb{R}^{kd \times kd}$, and $\otimes$ denotes the Kronecker product. This matrix $C$ satisfies  $\operatorname{rank}(C) = (k-1)d$, $C=C^{\top}$ and $C=C^2$ (see the Supplementary Material for a detailed derivation).

\section{Test statistics under the fixed-dimension regime}\label{section:test_stat}
In this section, we introduce test statistics for the hypothesis considered in  Section~\ref{sub:hypothesis}. 
These statistics are based on the estimators $\widehat{\theta}_i$ and their estimated joint covariance matrix \black{$\widehat{\Sigma}/n$}. In what follows, we describe how they can be used for inference, either via their asymptotic distributions or other resampling approaches such as the bootstrap, under suitable regularity conditions.

\subsection{Wald-type statistic}\label{sub:wts}
We consider the following WTS defined as $T_n = n \widehat{\theta}^{\top} C^{\top} \left(C \widehat{\Sigma} C^{\top}\right)^{-} C \widehat{\theta}$, where $(\cdot)^{-}$ denotes the Moore--Penrose inverse (see, e.g.,~\citeauthor{graybill1976linear}, \citeyear{graybill1976linear}, p. 24). This statistic is widely used in multivariate inference, as it provides a classical approach to testing general linear hypotheses about the underlying parameters. In particular, it is asymptotically distribution-free under suitable regularity conditions, as shown in the next proposition.
\begin{proposition}
Suppose that $C\Sigma C^{\top}$ has rank $(k-1)d$ and that the assumptions in  Lemma~\ref{Lemma:varconsist} hold. Then, under $H_0$, it follows that $T_n \tl U_{(k-1)d}$, where $U_{(k-1)d}$ has a $\chi^2_{(k-1)d}$ distribution.
\label{Proposition:WTS}
\end{proposition}
We write $\Sigma > 0$ if $\Sigma$ is positive definite and $\Sigma \ge 0$ if it is positive semidefinite. Notice that if $\Sigma > 0$, then $\operatorname{rank}\left(C\Sigma C^{\top}\right) = \operatorname{rank}(C) = (k-1)d$.
\begin{lemma}
Suppose that the assumptions of  Proposition~\ref{Proposition:WTS} hold and that $\Sigma>0$. Then, $ T_n/n \tp \Delta_T$, where $\Delta_T\geq 0$ and $\Delta_T=0$ if and only if $H_0$ is true.
\label{Lemma:consistencyWTS}
\end{lemma}
From  Lemma~\ref{Lemma:consistencyWTS} it is reasonable to reject $H_0$ for large values of $T_n$. From  Proposition~\ref{Proposition:WTS}, the test that rejects $H_0$ when $T_n > \chi^2_{(k-1)d, 1-\alpha}$,  $\chi^2_{(k-1)d, 1-\alpha}$ denoting the $(1-\alpha)$-quantile of the $\chi^2_{(k-1)d}$ distribution, has asymptotic level $\alpha$. 

In order to characterize the behavior of the test under departures from the null, we introduce a sequence of local alternatives.

\begin{assumption}[Local Alternatives]
$H_0$ does not hold and 
$\theta_n = \theta_0 + n^{-\gamma}b,$
for some $\gamma \ge 0$ and some vector 
$b = [b_i]_{i=1}^{kd}$ satisfying 
$C \theta_0 = 0$ and 
$C b \neq 0$, where $\theta_n=[\theta_{1,n}^{\top},\ldots,\theta_{k,n}^{\top}]^{\top}$, with $\theta_{i,n}=f(\eta_{i,n})$ and $\eta_{i,n}$ the value of $\eta_i$ under the distribution of $X_i$, which now depends on $n$. Moreover, $\sqrt{n_i}(\widehat{\theta}_i-\theta_{i,n}) \tl Z_i$, $Z_i \sim \mathcal{N}_d(0,\Sigma_i)$, and $\widehat{\Sigma}_i \tp \Sigma_i$, for some matrix $\Sigma_i$, $i \in \{1,\ldots,k\}$.
\label{assumption}
\end{assumption}

\begin{proposition}\label{Proposition:PowerWTS}
Let $\mathcal{P}=\lim_{n\to\infty}\Prob(T_n>\chi^2_{(k-1)d,1-\alpha})$, suppose that Assumption~\ref{assumption}, \eqref{eq:independent_sample} and \eqref{eq:compsamp} hold, and that $\Sigma>0$. Then,
\begin{equation*}
\mathcal{P} =
\begin{cases}
1,  & \text{if } \; 0 \leq \gamma < 1/2, \\
p_{W}(\alpha), & \text{if } \;\gamma = 1/2, \\
\alpha, & \text{if } \; \gamma > 1/2,
\end{cases}
\end{equation*}
where $\alpha < p_W(\alpha) < 1$.
\end{proposition}

As shown in  Proposition~\ref{Proposition:PowerWTS}, the proposed test exhibits desirable power properties.  However, in practice, computing $T_n$ requires inverting the matrix $C \widehat{\Sigma} C^\top$, which can be computationally intensive for large values of the parameter dimension $d$ and may be inaccurate if the sample sizes $n_i$ are not large enough or $d$ increases. For large $d$, as illustrated in the simulation study in  Section~\ref{sec:simulations}, the convergence of the empirical level to the nominal value requires large sample sizes. This fact has also been observed in another papers (see e.g.,~\citeauthor{Brunner1997}, \citeyear{Brunner1997}).

\begin{remark}
The WTS unifies the literature by encompassing the two-sample tests in \citet{Bhoj1993,OlkinFinn1995} and \citet{Davidson2009}, the multi-sample homogeneity statistic of \citet{Paul1989}, and the multivariate formulations of \citet{Jennrich1970} and \citet{Ditzhaus2025}.
\end{remark}

\subsection{ANOVA-type statistic}\label{sub:ats}
\label{subsection:ATS}

We consider the following ATS defined as $Q_n = n \widehat{\theta}^{\top} C \widehat{\theta}$. Under suitable regularity conditions, the ATS converges in distribution to a weighted linear combination of $\chi^2_1$ variables as shown in the next proposition.

\begin{proposition}
Suppose that the assumptions in  Lemma~\ref{Lemma:varconsist} hold. Then under $H_0$, it follows that 
$Q_n\tl Q:= \sum_{\ell=1}^{kd} \lambda_{\ell} Z^2_{\ell}$,
where the $\{Z^2_{\ell}\}_{\ell=1}^{kd}$ are independent random variables having a $\chi^2_1$ distribution and $\lambda_1=\lambda_1(C\Sigma C)\ge \cdots\ge \lambda_{kd}(C\Sigma C)=\lambda_{kd}$ denote the eigenvalues of $C \Sigma C$ arranged in nonincreasing order.  
\label{Proposition:ATS}
\end{proposition}

\begin{lemma}
Suppose that \eqref{eq:independent_sample} holds, that $\E\left\{\left\lvert  h(X_{i1},\ldots,X_{im})\right\rvert  \right\} < \infty$ and that $f$ is continuous at $\eta_i$, $i \in \{1,\ldots,k\}$. Then, $Q_n/n \tas \Delta_Q$,
with $\Delta_Q \geq 0$ and $\Delta_Q= 0$  if and only if $H_0$ is true.
\label{Lemma:consistencyATS}
\end{lemma}

As a consequence of  Lemma~\ref{Lemma:consistencyATS} it is reasonable to reject $H_0$ for large values of $Q_n$. From  Proposition~\ref{Proposition:ATS}, if $\lambda_1,\ldots,\lambda_{kd}$ were known, we could construct an asymptotic $\alpha$-level test. 
Let $c_{1-\alpha}$ denote the $(1-\alpha)$-quantile of  
$Q$. 
Then, the test that rejects $H_0$ when $Q_n > c_{1-\alpha}$, 
has asymptotic level $\alpha$. This test has desirable power properties, which are characterized in the next proposition.
\begin{proposition}\label{Proposition:PowerATS}
Let $\mathcal{P}=\lim_{n\to\infty}\Prob(Q_n>c_{1-\alpha})$. Suppose that Assumption~\ref{assumption}, \eqref{eq:independent_sample} and \eqref{eq:compsamp} hold. Then,
\begin{equation*}
\mathcal{P} =
\begin{cases}
1,  & \text{if } \; 0 \leq \gamma < 1/2, \\
p_A(\alpha), & \text{if } \;\gamma = 1/2, \\
\alpha, & \text{if } \; \gamma > 1/2,
\end{cases}
\end{equation*}
where $\alpha < p_A(\alpha) < 1$.
\end{proposition}
The result in  Proposition~\ref{Proposition:PowerATS} remains true if $c_{1-\alpha}$ is replaced with a consistent estimator.  As $\Sigma$ is unknown, and hence the eigenvalues of $C\Sigma C$, the distribution of $Q$ is also unknown. Next, we discuss two approximations of the null distribution of \black{$Q_n$}.

\subsubsection{Imhof's approximation}\label{Imhof}
This method consists in approximating the null distribution of $Q_n$ through an estimation of the distribution of $Q$.   Since the eigenvalues of $C \Sigma C$ are unknown, we approximate them using those of $C \widehat{\Sigma} C$. 
Let $\widehat{\lambda}_{\ell} = \lambda_{\ell}\left(C \widehat{\Sigma} C\right)$. From  Lemma~\ref{Lemma:varconsist} and the Continuous Mapping Theorem, we have that $C\widehat{\Sigma} C \tp C\Sigma C$. Applying Weyl's Perturbation Theorem (see Corollary III.2.6 in \citet{bhatia2013matrix}), it follows that
$
\widehat{\lambda}_{\ell} \tp \lambda_{\ell}$, for $1 \le \ell \le kd,
$
and hence from Slutsky's Theorem we have that
$
\sum_{\ell=1}^{kd} \widehat{\lambda}_{\ell} Z_{\ell}^2 \;\tl\; \sum_{\ell=1}^{kd} \lambda_{\ell} Z_{\ell}^2.
$
With this, the $p$-value can be estimated as follows,
\begin{equation*}
\widehat{p}_{\text{Imhof}} 
= \Prob_{*}\!\left(
\sum_{\ell=1}^{kd} \widehat{\lambda}_{\ell} Z_{\ell}^2 \;\ge\; Q_{n,\text{obs}}\right),
\end{equation*}
where $Q_{n,\text{obs}}$ is the value of $Q_n$ evaluated at the observed data, and $\Prob_{*}$ is the conditional probability given the data, which, from a practical point of view, means that $\widehat{\lambda}_1,\ldots, \widehat{\lambda}_{kd}$ are treated as fixed quantities (instead of random variables). In practice, $\widehat{p}_{\text{Imhof}}$ is evaluated numerically using Imhof's algorithm for quadratic forms in normal variables \citep{Imhof1961}, as implemented in the R package \texttt{CompQuadForm} \citep{Compquadform} through the function \texttt{imhof()}.

\subsubsection{Weighted bootstrap approximation}\label{sub:bootstrap}

As an alternative to the asymptotic approximation described above, we employ a weighted bootstrap (or multiplier bootstrap) procedure. This approach approximates the sampling distribution of $\widehat{\theta}$ by re-weighting the empirical influence functions, which are represented here by the centered jackknife pseudovalues.
\begin{equation}
\label{eq:random_weights}
\begin{minipage}{0.9\linewidth}
    Let $W_{i1},\ldots,W_{i n_i}$, $i \in \{1,\ldots, k\}$, be a sequence of independent and identically distributed (i.i.d.) random weights independent of the data $\mathcal{X}=\{\boldsymbol{X}_1, \ldots, \boldsymbol{X}_k\}$, with  $\mathbb{E}(W_{ij}) = 0$  and $\V(W_{ij}) = 1.$
\end{minipage}
\end{equation}
We define the bootstrap estimator of $\theta_i$ as $\widehat{\theta}_i^{*} = \widehat{\theta}_i + (1/n_i)
\sum_{j=1}^{n_i} W_{ij} \left(T_{i,j} - \overline{T}_i\right)$, where $T_{i,j}$ are the jackknife pseudovalues defined in  Section~\ref{subsec:varest} and $\overline{T}_i = n_i^{-1}\sum_{j=1}^{n_i} T_{i,j}$. 

Let $\mathbb{E}_*$ and $\V_*$ denote the expectation and variance with respect to $\Prob_{*}$, which from a practical point of view, means that $\widehat{\theta}_i, T_{i,1},\ldots,T_{i,n_i}$, $i \in \{1,\ldots,k\}$, are treated as fixed quantities. Consequently, it follows from \eqref{eq:random_weights} that $\mathbb{E}_*(\widehat{\theta}_i^{*}) = \widehat{\theta}_i$ and  $\V_*(\sqrt{n_i}\widehat{\theta}_i^{*})$ coincides  with the jackknife variance estimator. The null distribution of the ATS can be estimated through the conditional distribution of $Q_n^{*} = n (\widehat{\theta}^{*} - \widehat{\theta})^{\top} C (\widehat{\theta}^{*} - \widehat{\theta})$, given the data.  The validity of this approximation is established in the following result.
\begin{theorem}\label{theorem:bootstrap_consistency}
Suppose that \eqref{eq:independent_sample}, \eqref{ass:finitenessh2_bounzeta} and \eqref{eq:compsamp} hold, that $f$ has continuous second-order partial derivatives in a neighborhood of $\eta_i$, that $\Sigma_i>0$, $i\in\{1,\ldots,k\}$, and that the weights satisfy \eqref{eq:random_weights}. Then,
\begin{equation*}
\sup_{x \in \R}| \Prob_*(Q_n^{*}\le x)-\Prob(Q\le x)| \tp0,
\end{equation*} where $Q$ is defined as in  Proposition~\ref{Proposition:ATS}.
\end{theorem}

As an immediate consequence of  Theorem~\ref{theorem:bootstrap_consistency} and  Proposition~\ref{Proposition:ATS}, under $H_0$ we have that
\begin{equation*}
\sup_{x \in \R}| \Prob_*(Q_n^{*}\le x)-\Prob(Q_n\le x)| \tp0.
\end{equation*} Hence, the weighted bootstrap consistently approximates the null distribution of $Q_n$. With this, the $p$-value can be estimated as follows
\begin{equation*}
\widehat{p}_{\text{boot}}=\Prob_{*}\left(Q_{n}^{*} \geq Q_{n, \text{obs}}\right).
\end{equation*}
As usual, $\widehat{p}_{\text{boot}}$ is approximated by simulation.

\section{The increasing-dimension regime}\label{sub:incdim}

In this section, we approximate the distribution of the ATS in the regime where the dimension of $\theta$, $d=d(n)$, tends to infinity. The dimension $q$ of $\eta$ may also grow with the sample sizes, provided that the number of entries of each component of $f$ remains fixed. We do not provide the analogous derivation for the WTS since the WTS requires the inverse of the estimated asymptotic covariance matrix, which becomes unreliable as the dimension grows \citep{Bai1996}. Even when $\Sigma$ is known, the accuracy of the chi-squared approximation for quadratic forms in growing dimension deteriorates unless $d$ grows sufficiently slowly relative to $n$. Both facts are consistent with the severe type I error inflation exhibited by the WTS in  Section~\ref{numerical:experiments}. The following result provides the limiting distribution of the ATS test statistic.

\begin{theorem}
Suppose that $H_0$ holds, that \eqref{eq:independent_sample}, \eqref{ass:finitenessh2_bounzeta} and \eqref{eq:compsamp} hold,  that 
$f$ has continuous second-order partial derivatives in a neighborhood of $\eta_i$ and that
\begin{equation}
\max_{1\le r\le d}\;
\mathbb{E}\!\left\{\left\lvert  e_r^\top D_i M_h h_1(X_{i1})\right\rvert  ^4\right\}\leq M, \quad i \in \{1,\ldots,k\},
\label{eq:boundh4}
\end{equation}
for some constant $M>0$, where $e_r$ denotes the $r$th canonical vector of $\R^d$. Suppose also that
\begin{equation}
M_1 \;\leq\; \lambda_{\min}^{+}\!\bigl(C \Sigma C\bigr)
\leq \lambda_{1}\!\bigl(C\Sigma C\bigr)
\leq M_2,
\label{eigencontrol}
\end{equation}
where $\lambda_{\min}^{+}\!\bigl(C \Sigma C\bigr)$ and $\lambda_{1}\!\bigl(C \Sigma C\bigr)$ are respectively
the smallest nonzero and the largest eigenvalue of $C \Sigma C$. Then,
$({Q_n-\mu_n})/{\sigma_n} \tl Z$,
as $d\to \infty$, $d/n \to 0$, where $Z \sim \mathcal{N}(0,1)$, $\mu_n=\operatorname{tr}\!\bigl(\Sigma_n C\bigr)$ and $\sigma_n^2=2\,\operatorname{tr}\!\bigl\{(\Sigma_n C)^2\bigr\}$. The matrix $\Sigma_n$ is given by $\Sigma_n=\operatorname{diag}\left\{(n/n_1) \Sigma_1, \ldots, (n/n_k) \Sigma_k\right\}$, with $\Sigma_i$ defined in \eqref{eq:definitionsigmai}.

\label{theorem:assymptoticdinfty}
\end{theorem}

The next result shows that the asymptotic properties derived in  Theorem~\ref{theorem:assymptoticdinfty} remain valid when $\Sigma_n$ is replaced by $\widehat{\Sigma}$, under suitable regularity conditions.

\begin{proposition}
Suppose that the assumptions in  Theorem~\ref{theorem:assymptoticdinfty} hold, and that
\begin{equation}
\E\left\{\left\lvert  h^{(r)}(X_{i 1},\ldots,X_{i m_r})\right\rvert  ^4\right\}< M, \quad r \in \{1,\ldots, q\}, \; i \in \{1,\ldots,k\},
\label{eq:finitinessh4}
\end{equation}
for some constant $M>0$. Then, under $H_0$,
$
({Q_n-\widehat{\mu}_n})/{\widehat{\sigma}_n} \tl Z,
$
as $d\to \infty$, $d/n \to 0$, where $Z \sim \mathcal{N}(0,1)$, 
$\widehat{\mu}_n=\operatorname{tr}\left(\widehat{\Sigma} C\right)$, $\widehat{\sigma}^2=2\,\operatorname{tr}\!\left\{\left(\widehat{\Sigma} C\right)^2\right\}$. The matrix $\widehat{\Sigma}$ is given by $\widehat{\Sigma}=\operatorname{diag}\left\{(n/n_1) \widehat{\Sigma}_1, \ldots, (n/n_k) \widehat{\Sigma}_k\right\}$, with $\widehat{\Sigma}_i$ defined in  \eqref{eq:WidehatSigma}.
\label{prop:RatioConsistentHD}
\end{proposition}
From  Lemma~\ref{Lemma:consistencyATS}, it follows that evidence against $H_0$ is captured exclusively by the upper tail of the distribution, which leads to the critical region
\begin{equation}
{(Q_n - \widehat{\mu}_n)}/{\widehat{\sigma}_n} > z_{1-\alpha},
\label{eq:test ATS-ID}
\end{equation}
where $z_{1-\alpha}$ denotes the $(1-\alpha)$-quantile of the standard normal distribution.

\begin{remark}\label{rem:ATSIDfixed}
The standard normal approximation in \eqref{eq:test ATS-ID} is specific to the increasing-dimension regime. Under the fixed-dimension regime, since $C^2=C$ implies $\operatorname{tr}\left(\Sigma C\right)=\sum_{\ell=1}^{kd}\lambda_\ell$ and $\operatorname{tr}\left\{\left(\Sigma C\right)^2\right\}=\sum_{\ell=1}^{kd}\lambda_\ell^2$, with $\lambda_1,\ldots,\lambda_{kd}$ as in Proposition~\ref{Proposition:ATS}, combining Proposition~\ref{Proposition:ATS}, Lemma~\ref{Lemma:varconsist}, the Continuous Mapping Theorem and Slutsky's Theorem yields, under $H_0$, $(Q_n-\widehat{\mu}_n)/\widehat{\sigma}_n \tl \widetilde{Q}:=\bigl(Q-\sum_{\ell=1}^{kd}\lambda_\ell\bigr)\big/\bigl(2\sum_{\ell=1}^{kd}\lambda_\ell^2\bigr)^{1/2}$. The limit $\widetilde{Q}$ has zero mean and unit variance, but, under the fixed dimension regime, its third central moment,  $m_3=8\sum_{\ell=1}^{kd}\lambda_\ell^3\big/\bigl(2\sum_{\ell=1}^{kd}\lambda_\ell^2\bigr)^{3/2}$, is strictly positive under \eqref{eigencontrol}. Then for fixed $d$, the test in \eqref{eq:test ATS-ID} does not have asymptotic level $\alpha$. This agrees with the mildly liberal behavior observed in the simulations of  Section~\ref{numerical:experiments} when $kd$ is small. Notice that under \eqref{eigencontrol}, we have that $m_3 \leq M/(kd)^{1/3} \to 0$, as $d \to \infty$. 
\end{remark}

\begin{remark}\label{rem:finite_k}
\textcolor{black}{
Theorem~\ref{theorem:assymptoticdinfty} is stated for fixed $k$ and $d/n \to 0$ with
$d, n\to\infty$. In this setting, the effective dimension of the quadratic form is
of order $kd$. Thus, for fixed $k$, conditions written in terms of $d$ are equivalent,
up to constants, to conditions written in terms of $kd$.
This matters when interpreting simulations with varying $k$. For example, in
balanced designs, $d/n=d/(kn_i)$ decreases with $k$, while the effective
dimension $kd$ increases. Hence, $d/n$ and $d$ alone may be misleading; the ratio
$kd/n$ and the effective dimension $kd$ are more informative finite-sample indicators for the validity of the normal approximation.
This does not extend the theorem to $k\to\infty$. }
\end{remark}

\section{Numerical experiments}\label{numerical:experiments}
This section summarizes the results of several numerical experiments. We first evaluate the level and the power with synthetic data, and then with real data. The results from the real data are collected in the Supplementary Material. Finally, we give some practical guidelines.

\subsection{Simulation study}\label{sec:simulations}
We evaluated the finite-sample performance of the proposed tests through extensive Monte Carlo simulations. We present here a summary of the main findings. See the Supplementary Material for a more detailed description.

By abuse of notation, we denote by WTS the test based on the WTS statistic, by ATS the test based on the ATS with $p$-value approximated through $\widehat{p}_{\text{Imhof}}$, and by ATS-ID the test defined in \eqref{eq:test ATS-ID}. Finally, we denote by \black{WBS-N, WBS-R, and WBS-M the ATS-based tests with $p$-values approximated through $\widehat{p}_{\text{boot}}$ described in  Section~\ref{sub:bootstrap}, with random 
standard Normal weights, Rademacher weights and Mammen weights \citep{Mammen1993}, respectively. When referring to these three bootstrap-based tests collectively, we simply use the acronym WBS.}

\subsubsection{Univariate inference: Variance and Gini index}

We first applied the methodology to two univariate parameters: the variance and the Gini index.

To assess the type I error (see Table S1 for the variance and Table S3 for the Gini index in the Supplementary Material), data were generated under four distributional configurations.  As the number of groups $k$ increases, the  WTS becomes liberal, showing inflated type I error rates, particularly for small sample sizes. In contrast, the ATS and WBS provide a better control of the type I error. Although they exhibit a conservative behavior for small sample sizes and large $k$, their empirical levels converge to the nominal value as sample sizes increase. In the variance setting, this conservativeness is more
pronounced for skewed or heavy-tailed distributions (see Table S1, Cases 2--4 in the Supplementary Material). \black{The empirical results confirm that WBS-N, WBS-R, and WBS-M are asymptotically equivalent, with all empirical levels properly converging to the nominal level for large sample sizes. However, noticeable differences emerge for smaller sample sizes. WBS-R systematically exhibits the least conservative behavior, closely followed by WBS-M and WBS-N. The same pattern is observed across all simulation studies.}

To analyze the power (see Table S2 for the variance and Table S4 for the Gini index in the Supplementary Material), we simulated three distinct patterns of deviation from the null hypothesis. As expected, for all tests and scenarios, the empirical power increases monotonically with the sample size and the magnitude of the deviation. The WTS generally exhibits the highest empirical power, but this result must be interpreted with caution due to its inflated type I error rates. The ATS and WBS offer the most reliable trade-off between size control and power.

\subsubsection{Multivariate inference: Means and covariance matrices}

\textcolor{black}{We extended the analysis to the joint equality of mean vectors and covariance
matrices. The main multivariate simulation study considered dimensions
$d\in\{5,20,65\}$, whereas an additional type I error and computational comparison
was carried out for larger dimensions $d\in\{135,299\}$.}

\textcolor{black}{To assess the type I error (see Tables S5 and S7 in the Supplementary Material), we generated data
from multivariate normal distributions under three covariance structures. The
results reveal a sharp contrast between the methods driven by the dimensionality.
The WTS suffers a severe type I error inflation as $d$ increases. In contrast, the ATS, WBS and ATS-ID provide more stable alternatives, although
they tend to become conservative for larger values of $kd/n$. Their empirical
levels approach the nominal levels at comparable rates as $kd/n$
decreases.
ATS-ID additionally requires the effective dimension $kd$ to be large
enough for the increasing-dimension normal approximation to be reliable; when
$kd$ is too small, it may display a mildly liberal behavior. This behavior is
consistent with  Remarks~\ref{rem:ATSIDfixed} and \ref{rem:finite_k}. When the increasing-dimension approximation underlying ATS-ID is reliable, the differences among ATS, WBS and ATS-ID are mainly computational. The ATS
requires the spectral decomposition associated with the weighted chi-square
calibration, the WBS relies on repeated resampling evaluations, and ATS-ID
is based on trace-type quantities. Thus, in such regimes, the main
appeal of ATS-ID is computational rather than a systematic improvement in type I
error control, especially when the spectral decomposition required by the ATS
becomes computationally burdensome in practice. This is reflected in
Table~\ref{table:ComputationTime}.}

\begin{table}[ht] \black{
\caption{CPU time, in seconds, required to compute 10 $p$-values for the joint test of mean vectors and covariance matrices under the multivariate standard normal configuration. For WBS, the reported value is the average over WBS-N, WBS-R and WBS-M.}\label{table:ComputationTime}
\centering
\setlength{\tabcolsep}{3pt}
\resizebox{\textwidth}{!}{
\begin{tabular}{cc cccc cccc cccc}
\toprule
 & 
& \multicolumn{4}{c}{$n_i=100$}
& \multicolumn{4}{c}{$n_i=1000$}
& \multicolumn{4}{c}{$n_i=10000$} \\
\cmidrule(lr){3-6}
\cmidrule(lr){7-10}
\cmidrule(lr){11-14}
$d$ & $k$ & WTS & ATS & WBS & ATS-ID 
  & WTS & ATS & WBS & ATS-ID
  & WTS & ATS & WBS & ATS-ID \\
\cmidrule(lr){1-1}
\cmidrule(lr){2-2}
\cmidrule(lr){3-3}
\cmidrule(lr){4-4}
\cmidrule(lr){5-5}
\cmidrule(lr){6-6}
\cmidrule(lr){7-7}
\cmidrule(lr){8-8}
\cmidrule(lr){9-9}
\cmidrule(lr){10-10}
\cmidrule(lr){11-11}
\cmidrule(lr){12-12}
\cmidrule(lr){13-13}
\cmidrule(lr){14-14}
\multirow{3}{*}{135}
& 2  & 0.45 & 0.14 & 0.15 & 0.05 & 0.53 & 0.23 & 1.17 & 0.14 & 1.76 & 1.51 & 11.86 & 1.40 \\
& 5  & 5.83 & 1.11 & 0.33 & 0.08 & 6.48 & 1.34 & 2.86 & 0.35 & 9.72 & 4.51 & 29.38 & 3.41 \\
& 10 & 46.04 & 6.98 & 0.71 & 0.15 & 50.71 & 7.73 & 5.88 & 0.78 & 57.60 & 14.14 & 58.56 & 6.99 \vspace{3pt} \\
\multirow{3}{*}{299}
& 2  & 3.71 & 0.93 & 0.27 & 0.08 & 4.42 & 1.32 & 2.22 & 0.58 & 9.81 & 6.57 & 21.29 & 5.86 \\
& 5  & 56.87 & 9.24 & 0.66 & 0.19 & 66.40 & 11.01 & 5.34 & 1.46 & 80.54 & 23.95 & 52.90 & 14.56 \\
& 10 & 454.95 & 69.81 & 1.28 & 0.47 & 538.66 & 71.92 & 10.85 & 3.01 & 558.19 & 99.20 & 104.32 & 29.22 \\
\bottomrule
\end{tabular}
}
}
\end{table}

Finally, to analyze the power (see Table S6 in the Supplementary Material), we simulated four specific multivariate alternatives: a location shift, a scale shift, a dependence structure shift and a joint location-covariance shift. As expected, for all tests and scenarios, the empirical power increases monotonically with the sample size. The WTS generally exhibits the highest empirical power, although this superiority is misleading due to its lack of type I error control. The ATS, WBS, and ATS-ID demonstrate reliable power properties. Notably, their performance depends on the signal structure relative to the dimension. For sparse signals such as the location alternative affecting only one component, the power decreases as $d$ increases due to the accumulation of noise. Conversely, for global structural changes such as scale, dependence, and joint alternatives, the power increases with the dimension $d$, driven by the accumulation of signal across components.

\subsection{\textcolor{black}{Practical guidelines}}\label{practical:guidelines}

\textcolor{black}{
The numerical experiments showed that no single procedure uniformly dominates the
others across all settings. Instead, their practical behavior depends on the effective dimension $kd$, the quotient $kd/n$, and the
distributional features of the data.}

\textcolor{black}{
First, the WTS should be used with caution due to its
liberal behavior  in univariate
settings when  $k$ increases, which becomes especially severe
in multivariate settings as both $k$ and $d$ grow.}

\textcolor{black}{
In contrast, the ATS, WBS and ATS-ID provide more reliable type I error control. Although the ATS and WBS are motivated by fixed-dimensional
calibrations, the simulations show that they remain competitive with ATS-ID when
the dimension increases. These three
calibrations approach the nominal level at comparable rates as the available
sample size increases relative to the effective dimension of the contrast. In
our simulations, satisfactory behavior was typically observed when $kd/n$ was
small enough, roughly below 0.03, although this value should be understood as
an empirical guideline that may depend on the underlying distribution and the
specific testing problem.}

\textcolor{black}{
It is important to note that the ATS-ID additionally requires the effective
dimension of the contrast to be large enough for the increasing-dimension normal
approximation to be reliable. In our simulations, this approximation became
effective for sufficiently large values of $kd$, roughly from the order of 600
onward. This value should again be interpreted as an empirical guideline, since
it may depend on the underlying distribution and the specific testing problem.
For smaller effective dimensions, ATS-ID may show a mild liberal tendency, in agreement with  Remark~\ref{rem:ATSIDfixed}.}

\textcolor{black}{
Among the bootstrap procedures, WBS-R consistently exhibited the least
conservative behavior, followed by WBS-M and WBS-N. Although WBS-R may be
slightly liberal in some low-dimensional configurations, it is often preferable
in larger-$k$ and larger-$d$ regimes, where the dominant finite-sample issue is
conservativeness.}

\textcolor{black}{
Since the ATS, WBS and ATS-ID exhibit broadly similar type I error behavior in
the regimes where the ATS-ID is reliable,
computational efficiency becomes an important criterion for choosing among them. In such cases  ATS-ID is the best choice.}

\textcolor{black}{
Finally, the simulations and empirical applications indicate that distributional
features may substantially affect the finite-sample behavior of the procedures.
In particular, skewness, heavy tails, and strong scale imbalances among
variables tend to induce a stronger conservativeness of ATS and WBS procedures,
as observed in the univariate variance simulations and in the real data
application (see Table S1, Cases 2--4, and Table S9 in the Supplementary Material). In practice,
when suitable transformations are available they may considerably improve the empirical calibration of the
tests (see e.g.,~Section S5.2 in the Supplementary Material).}

\section{Conclusions and future work}\label{section:conclusions}

In this paper, a general and unified framework for testing the equality of parameters via $U$-statistics across multiple populations has been developed. Two  statistics have been considered: a Wald-type test statistic  and an ANOVA-type  test statistic. A key feature of the presented methodology is that the covariance matrices involved in these statistics are consistently estimated using the jackknife estimator, enabling  asymptotically exact inference without relying on specific parametric assumptions. Extensive simulations were conducted to evaluate and compare these procedures. The results indicate that, while the procedures perform well in standard settings ($d$, $k$ fixed, large $n$), their behavior differs as the dimension increases. Specifically, the WTS tends to be excessively liberal in high-dimensional scenarios, whereas the asymptotic ATS and WBS tend to be conservative. 
When $kd/n$ is small and $kd$ is large,  ATS, WBS and ATS-ID give levels close to the nominal values, but ATS-ID is computationally more efficient.

A number of extensions are possible. For instance, the current theory relies on the independence between samples as stated in \eqref{eq:independent_sample}. It would be relevant to extend this framework to accommodate dependent data. Additionally, another pertinent extension would be the adaptation of these methods to handle missing data. In this direction,  \citet{aleksic2025two} recently proposed weighting and imputation approaches for two-sample testing via energy distance. Generalizing their strategies to the $U$-statistics framework represents a promising avenue for future research.

\section{Proofs}\label{section:proofs}
Throughout this section, $\| \cdot\| _2$ denotes the Euclidean norm and $\| \cdot\| $ represents the operator norm for matrices induced by $\| \cdot\| _2$. The symbols $M, M_1, M_2, \dots$ denote generic positive constants that are not important and may vary from one step to another, or even between different occurrences within the same line of an inequality. A sequence of random variables is $o_{\Prob}(1)$ if it converges to 0 in probability, and is $O_{\Prob}(1)$ if it is bounded in probability.

Before proving  Lemma~\ref{Lemma:varconsist}, we first state a preliminary result concerning the jackknifed statistic.
This result can be found in the proof of Theorem 5 in \citet{Arvesen1969}. Here we state and provide a detailed proof of the result.

\begin{lemma}
Let $X_1, \dots, X_n$ be a random sample and let $\widehat{\eta} \in \R$ be a U-statistic of order $m$ with kernel $h$ such that 
$\mathbb{E}\left\{| h(X_1,\dots,X_m)| ^2\right\} < \infty$, $\V\{h_1(X_1)\}=\zeta_1>0$ and without loss of generality consider that 
$\mathbb{E}\left\{h(X_1,\dots,X_m)\right\} = 0$. Denote by 
$\widehat{\eta}_{-j}$ the jackknife version leaving out observation $X_j$. Then
$\max_{1\le j \le n} \big| \widehat{\eta}_{-j} - \widehat{\eta}\big|  = O_{\Prob}\left(n^{-1/2}\right)$.

\label{lemma:etaj-eta}
\end{lemma}
\begin{Proof}
First, for each $j \in \{1,\ldots, n\}$, we have that
$\widehat{\eta}_{-j}=\binom{n}{m}\binom{n-1}{m}^{-1} \widehat{\eta}-mY_j/(n-m)$,
where
\begin{equation*}
Y_j=\binom{n-1}{m-1}^{-1}\sum_{\substack{1 \le i_1 < \cdots < i_{m-1} \le n, \, \,  i_1,\dots,i_{m-1} \neq j}} h(X_j,X_{i_1},\ldots,X_{i_{m-1}}),
\end{equation*}
can be seen as a $U$-statistic conditioned on $X_j$ with kernel $g(x_1,\ldots,x_{m-1})=h(X_j,x_1,\ldots,x_{m-1})$.  Now, applying the total variance formula we get that
\begin{equation}
\V(Y_j)=\V\{\E(Y_{j} \vert X_{j})\} + \E\{\V(Y_{j} \vert X_{j})\}.
\label{eq:vartotal}
\end{equation}
For the first term we have that $\V\{\E(Y_{j} \vert X_{j})\}=\zeta_1$. For the second we can apply equation \black{(S3) in the Supplementary Material}, and then we have that
\begin{equation}
\E\{\V(Y_{j} \vert X_{j})\}=\binom{n-1}{m-1}^{-1} \sum_{c=1}^{m-1} \binom{m-1}{c} \binom{n-m}{m-1-c} \E[\V\{h_{c+1}(X_j,X_1,\ldots,X_c)\vert  X_j\}].
\label{eq:varcond}
\end{equation}
Now, applying the total variance law, it follows that
\begin{equation*}
\begin{aligned}
\zeta_{c+1} &= \V\{h_{c+1}(X_j,X_1,\dots,X_c)\}
= \E[\V\{h_{c+1}(X_j,X_1,\dots,X_c) \mid X_j\}] \\
&\quad + \V[\E\{h_{c+1}(X_j,X_1,\dots,X_c) \mid X_j\}]
\ge \E[\V\{h_{c+1}(X_j,X_1,\dots,X_c) \mid X_j\}].
\end{aligned}
\end{equation*}
With this, recalling \eqref{eq:vartotal} and \eqref{eq:varcond} and that $\zeta_1\leq\zeta_m=\mathbb{E}\left\{| h(X_1,\dots,X_m)| ^2\right\} < \infty$ we get that
$\V(Y_j)=\zeta_1 + O(1/n)$. From this, since the random variables $Y_1,\ldots,Y_n$ are exchangeable, we have that
\begin{equation*}
\begin{aligned}
\Prob\!\left(\max_{1 \le j \le n} \sqrt{n}\frac{m}{\,n - m\,}\,| Y_j|  > \varepsilon \right)
&\le \sum_{j=1}^n \Prob\!\left( \sqrt{n}\frac{m}{\,n - m\,}\,| Y_j|  > \varepsilon \right)    =n\,\Prob\!\left( \sqrt{n}\frac{m}{\,n - m\,}\,| Y_1|  > \varepsilon \right)\\  
&\le {m^2 n^2\,\zeta_1}/\{{(n - m)^2\,\varepsilon^2}\} + O(1/n) \leq M/\varepsilon^2,
\end{aligned}
\end{equation*}
for large enough $n$. Therefore, $\max_{j \in \{1,\ldots,n\}}m | Y_j| /(n-m)=O_{\Prob}(n^{-1/2})$. With this, we get the result as follows  $\max_{j \in \{1,\ldots,n\}} \big| \widehat{\eta}_{-j} - \widehat{\eta}\big| \leq \max_{j \in \{1,\ldots,n\}}m/(n-m)| Y_j|  + m/(n-m)\left\lvert   \widehat{\eta}\right\rvert  =O_{\Prob}(n^{-1/2})$,
where we have applied $U$-statistics CLT \eqref{eq:assymptheta}, which yields $\left\lvert   \widehat{\eta}\right\rvert  =O_{\Prob}(n^{-1/2})$.
\end{Proof}

\begin{Prooflx}{Lemma:varconsist}
First, for simplicity we drop the subscript $i$ so that $\Sigma$, $n$, $\theta$ and $\eta$ stand for a generic $\Sigma_i$, $n_i$, $\theta_i$ and $\eta_i$ respectively. Let $\Sigma=\left(\sigma_{ab}\right)$ and $\widehat{\Sigma}=\Bigl(\widehat{\sigma}_{ab}\Bigr)$, $a,b \in \{1,\ldots, d\}$. Then, we will show that $\widehat{\sigma}_{ab} \tp \sigma_{ab}$.

First, note that since $f$ has continuous first partial derivatives, for each $j \in \{1,\ldots,n\}$, we have that
\begin{equation*}
f^{(r)}\left(\widehat{\eta}^{(1)}_{-j},\ldots, \widehat{\eta}^{(q)}_{-j}\right)= f^{(r)}\left(\widehat{\eta}^{(1)},\ldots, \widehat{\eta}^{(q)}\right)+\sum_{\ell=1}^q \left(\widehat{\eta}^{(\ell)}_{-j}-\widehat{\eta}^{(\ell)}\right)f^{(r)}_{\ell}(\tau_j), \quad r \in \{1,\ldots,d\},
\end{equation*}
or equivalently $\widehat{\theta}^{(r)}_{-j}= \widehat{\theta}^{(r)}+\sum_{\ell=1}^q \left(\widehat{\eta}^{(\ell)}_{-j}-\widehat{\eta}^{(\ell)}\right)f^{(r)}_{\ell}(\tau_j)$, $r \in \{1,\ldots,d\}$,
where $\tau_j$ (depending also on $r$) lies between $\left(\widehat{\eta}^{(1)},\ldots, \widehat{\eta}^{(q)}\right)$ and $\left(\widehat{\eta}^{(1)}_{-j},\ldots, \widehat{\eta}^{(q)}_{-j}\right)$.
Thus, from \eqref{eq:WidehatSigma}, we can write
\allowdisplaybreaks
\begin{align*}
\widehat{\sigma}_{ab} 
&= \frac{1}{n-1} \sum_{j=1}^{n} 
\left(T_{j}^{(a)}-\overline{T}^{(a)}\right)
\left(T_{j}^{(b)}-\overline{T}^{(b)}\right) = (n-1)\sum_{j=1}^{n} 
\left(\widehat{\theta}_{-j}^{(a)}-n^{-1}\sum_{r=1}^n \widehat{\theta}_{-r}^{(a)} \right)
\left(\widehat{\theta}_{-j}^{(b)}-n^{-1}\sum_{s=1}^n \widehat{\theta}_{-s}^{(b)} \right) \notag\\
&= (n-1)\sum_{j=1}^{n} 
\left\{ f^{(a)}\left(\widehat{\eta}^{(1)}_{-j},\dots, \widehat{\eta}^{(q)}_{-j}\right)
- n^{-1}\sum_{r=1}^{n} f^{(a)}\left(\widehat{\eta}^{(1)}_{-r},\dots, \widehat{\eta}^{(q)}_{-r}\right) \right\} \notag\\
&\quad \times 
\left\{ f^{(b)}\left(\widehat{\eta}^{(1)}_{-j},\dots, \widehat{\eta}^{(q)}_{-j}\right)
- n^{-1}\sum_{s=1}^{n} f^{(b)}\left(\widehat{\eta}^{(1)}_{-s},\dots, \widehat{\eta}^{(q)}_{-s}\right) \right\}. \notag
\end{align*}
Noting that for $\ell \in \{1,\ldots,q\}$, it follows that $\sum_{r=1}^n (\widehat{\eta}^{(\ell)}_{-r}-\widehat{\eta}^{(\ell)})=0$, we can write
\begin{align*}
\widehat{\sigma}_{ab}&= (n-1) \sum_{j=1}^{n} 
\left[
\sum_{\ell=1}^q \left(\widehat{\eta}^{(\ell)}_{-j}-\widehat{\eta}^{(\ell)}\right) f^{(a)}_{\ell}(\eta) 
+ \sum_{\ell=1}^q \left(\widehat{\eta}^{(\ell)}_{-j}-\widehat{\eta}^{(\ell)}\right) \left\{ f^{(a)}_{\ell}(\tau_j) - f^{(a)}_{\ell}(\eta) \right\} \right. \notag\\
&\quad \left. - n^{-1} \sum_{r=1}^n \sum_{\ell=1}^q \left(\widehat{\eta}^{(\ell)}_{-r}-\widehat{\eta}^{(\ell)}\right) \left\{ f^{(a)}_{\ell}(\tau_r) - f^{(a)}_{\ell}(\eta) \right\}
\right]  \left[
\sum_{t=1}^q \left(\widehat{\eta}^{(t)}_{-j}-\widehat{\eta}^{(t)}\right) f^{(b)}_t(\eta) \right.    \notag\\
&\quad \left. 
+ \sum_{t=1}^q \left(\widehat{\eta}^{(t)}_{-j}-\widehat{\eta}^{(t)}\right) \left\{ f^{(b)}_t(\tau_j) - f^{(b)}_t(\eta) \right\} 
- n^{-1} \sum_{s=1}^n \sum_{t=1}^q \left(\widehat{\eta}^{(t)}_{-s}-\widehat{\eta}^{(t)}\right) \left\{ f^{(b)}_t(\tau_s) - f^{(b)}_t(\eta) \right\}
\right] \notag\\
&=(n-1)\sum_{j=1}^n \left[\left(\widehat{\eta}_{-j}-\widehat{\eta}\right)^{\top} \nabla f^{(a)}(\eta)+ \left(\widehat{\eta}_{-j}-\widehat{\eta}\right)^{\top} \nabla \left\{f^{(a)}(\tau_j)-f^{(a)}(\eta)\right\} \right. \notag\\
&\quad \left. - n^{-1} \sum_{r=1}^n \left(\widehat{\eta}_{-r}-\widehat{\eta}\right)^{\top} \nabla\left\{f^{(a)}(\tau_r)-f^{(a)}(\eta)\right\}  \right]  \left[\left(\widehat{\eta}_{-j}-\widehat{\eta}\right)^{\top} \nabla f^{(b)}(\eta) \right.\\
&\quad \left. + \left(\widehat{\eta}_{-j}-\widehat{\eta}\right)^{\top} \nabla \left\{f^{(b)}(\tau_j)-f^{(b)}(\eta)\right\} 
- n^{-1} \sum_{s=1}^n \left(\widehat{\eta}_{-s}-\widehat{\eta}\right)^{\top} \nabla\left\{f^{(b)}(\tau_s)-f^{(b)}(\eta)\right\}  \right].
\end{align*}
The previous expression contains terms of five different forms:
\begin{align}
& (n-1)\sum_{j=1}^n (\widehat{\eta}_{-j}-\widehat{\eta})^{\top}  \nabla f^{(a)}(\eta) \;\, (\widehat{\eta}_{-j}-\widehat{\eta})^{\top} \nabla f^{(b)}(\eta), \label{eq:type1}\\
&  (n-1) \sum_{j=1}^{n}  (\widehat{\eta}_{-j}-\widehat{\eta})^{\top} \nabla f^{(a)}(\eta) \;\, (\widehat{\eta}_{-j}-\widehat{\eta})^{\top} \nabla\left\{f^{(b)}(\tau_j)-f^{(b)}(\eta)\right\}, \label{eq:type2}\\
&  (n-1) n^{-1} \sum_{j=1}^{n} \sum_{s=1}^n   (\widehat{\eta}_{-j}-\widehat{\eta})^{\top} \nabla f^{(a)}(\eta) \;\,   (\widehat{\eta}_{-s}-\widehat{\eta})^\top  \nabla \left\{ f^{(b)}(\tau_s) - f^{(b)}(\eta) \right\},  \label{eq:type3}\\
&  (n-1) \sum_{j=1}^n 
(\widehat{\eta}_{-j}-\widehat{\eta})^{\top} \nabla\left\{ f^{(a)}(\tau_j) - f^{(a)}(\eta) \right\} \;\,  (\widehat{\eta}_{-j}-\widehat{\eta})^{\top}  \nabla\left\{ f^{(b)}(\tau_j) - f^{(b)}(\eta) \right\}, \label{eq:type4}\\
&  (n-1)n^{-1} \sum_{r=1}^n \sum_{s=1}^n (\widehat{\eta}_{-r}-\widehat{\eta})^{\top} \nabla\left\{ f^{(a)}(\tau_r) - f^{(a)}(\eta) \right\} \;\,  (\widehat{\eta}_{-s}-\widehat{\eta})^{\top}  \nabla\left\{ f^{(b)}(\tau_s) - f^{(b)}(\eta) \right\}, \label{eq:type5}
\end{align}
where, in the last equation, we have already performed the summation over $j$. The remainder of the proof will show that term \eqref{eq:type1} converges in probability to $\sigma_{ab}$, while terms \eqref{eq:type2}--\eqref{eq:type5} converge to zero in probability. We begin with term \eqref{eq:type1}. This term can be written as
\begin{equation}
\left(n-1\right) \sum_{j=1}^n \sum_{\ell=1}^q \sum_{t=1}^q \left(\widehat{\eta}^{(\ell)}_{-j}-\widehat{\eta}^{(\ell)}\right)\left(\widehat{\eta}^{(t)}_{-j}-\widehat{\eta}^{(t)}\right) f^{(a)}_{\ell}\left(\eta\right) f^{(b)}_t\left(\eta\right)= \sum_{\ell=1}^q \sum_{t=1}^q Y_{\ell,t} \, f^{(a)}_{\ell}\left(\eta\right) f^{(b)}_t\left(\eta\right),
\label{eq:type1sums}
\end{equation}
where $Y_{\ell,t}=(n-1) \sum_{j=1}^n \left(\widehat{\eta}^{(\ell)}_{-j}-\widehat{\eta}^{(\ell)}\right)\left(\widehat{\eta}^{(t)}_{-j}-\widehat{\eta}^{(t)}\right)$.
As shown in Equation (36) of \citet{Arvesen1969}, $Y_{\ell,t}$ can be expressed as follows
\begin{equation}
Y_{\ell,t}= (n-1) n^{-1} \binom{n-1}{m_{\ell}}^{-1} \binom{n-1}{m_t}^{-1} \sum_{c=0}^{m_{\ell,t}}(cn-m_{\ell} m_t)  \binom{n}{m_{\ell}} \binom{m_{\ell}}{c} \binom{n-m_{\ell}}{m_t-c} U^{(\ell,t)}_c,
\label{eq:type1dev}
\end{equation}
with $m_{\ell,t}=\min\{m_{\ell},m_t\}$, $\ell,t \in \{1,\ldots,q\}$, and
\begin{equation*}
\begin{aligned}
U^{(\ell,t)}_c&= \left\{\binom{n}{m_{\ell}} \binom{m_{\ell}}{c} \binom{n-m_{\ell}}{m_t-c}\right\}^{-1}\sum \widetilde{h}^{(\ell)}(X_{\alpha_1},\ldots,X_{\alpha_c},X_{\beta_1},\ldots,X_{\beta_{m_{\ell}-c}}) \\
    & \quad \times \widetilde{h}^{(t)}(X_{\alpha_1},\ldots,X_{\alpha_c},X_{\gamma_1},\ldots,X_{\gamma_{m_t-c}}),
\end{aligned}
\end{equation*}
where the sum in the above expression is taken over all disjoint sets $\{\alpha_1,\ldots,\alpha_c\}$, $\{\beta_1,\ldots,\beta_{m_{\ell}-c}\}$, $\{\gamma_1,\ldots,\gamma_{m_t-c}\}$ of distinct integers chosen from $\{1,\ldots,n\}$. Notice that $U^{(\ell,t)}_c$ is a $U$-statistic with symmetric kernel
\begin{equation}
\begin{aligned}
K^{(\ell,t)}_c&(X_{\alpha_1},\ldots,X_{\alpha_c},X_{\beta_1},\ldots,X_{\beta_{m_{\ell}-c}},X_{\gamma_1},\ldots,X_{\gamma_{m_t-c}}) \\
&=\binom{m_{\ell}+m_t-c}{c(m_{\ell}-c)(m_t-c)}^{-1} \sum_{\Pi_m} \widetilde{h}^{(\ell)}(X_{\alpha_1},\ldots,X_{\alpha_c},X_{\beta_1},\ldots,X_{\beta_{m_{\ell}-c}}) \\
 & \quad \times \widetilde{h}^{(t)}(X_{\alpha_1},\ldots,X_{\alpha_c},X_{\gamma_1},\ldots,X_{\gamma_{m_t-c}}),
\end{aligned}
\label{eq:Kc}
\end{equation}
where $\sum_{\Pi_m}$ indicates the sum over all the $\binom{m_{\ell}+m_t-c}{c(m_{\ell}-c)(m_t-c)}$ possible permutations of the set of integers $\{\alpha_1,\ldots,\alpha_c,\beta_1$,  $\ldots,\beta_{m_{\ell}-c}$,  $\gamma_1,\ldots,\gamma_{m_t-c}\}$ that are not permutations only within the sets $\{\alpha_1,\ldots,\alpha_c\}$, $\{\beta_1,\ldots,\beta_{m_{\ell}-c}\}$ or $\{\gamma_1,\ldots, \gamma_{m_t-c}\}$.
Now, by applying the Cauchy--Schwarz inequality and recalling \eqref{ass:finitenessh2_bounzeta}, we get that
\begin{equation*}
\E\left(\left\lvert  K_c^{(\ell,t)}\right\rvert  \right) \leq \E^{1/2}\left\{\left\lvert  \widetilde{h}^{(\ell)}(X_1,\ldots,X_{m_{\ell}})\right\rvert  ^2\right\} \E^{1/2}\left\{\left\lvert  \widetilde{h}^{(t)}(X_1,\ldots,X_{m_t})\right\rvert  ^2\right\} < \infty,
\end{equation*} 
and hence, applying the SLLN for $U$-statistics (see Theorem 4 in \citeauthor{Arvesen1969}, \citeyear{Arvesen1969}), we get that
\begin{equation}
U_c^{(\ell,t)} \tp \E\left(K^{(\ell,t)}_c\right) \quad c \in \{0,\ldots, m_{\ell,t}\},\, \ell,t \in \{1,\ldots,q\},
\label{eq:SLLNUC}
\end{equation}
which means that  each $U_c^{(\ell,t)}=O_{\Prob}(1)$. In particular, applying this reasoning to \eqref{eq:type1dev}, we obtain that
\begin{equation}
Y_{\ell,t}=(n-1)  \binom{n-1}{m_{\ell}}^{-1} \binom{n-1}{m_t}^{-1}  \binom{n}{m_{\ell}}  \binom{n-m_{\ell}}{m_t-1} m_{\ell}U^{(\ell,t)}_1+o_{\Prob}(1),
\label{eq:YltSLNN}
\end{equation}
where the coefficient satisfies
\begin{equation}
(n-1)  \binom{n-1}{m_{\ell}}^{-1} \binom{n-1}{m_t}^{-1}  \binom{n}{m_{\ell}}  \binom{n-m_{\ell}}{m_t-1} m_{\ell} \to m_{\ell} m_t.
\label{eq:coefficientmlmt}
\end{equation}
Finally, from \eqref{eq:type1sums}, \eqref{eq:SLLNUC} and \eqref{eq:coefficientmlmt}  together with the fact that $\E\left(K^{(\ell,t)}_1\right)=\E\left\{\widetilde{h}_1^{(\ell)}(X_1)\widetilde{h}_1^{(t)}(X_1)\right\}$, we obtain that \eqref{eq:type1} converges in probability to $\sum_{\ell=1}^q \sum_{t=1}^q  f^{(a)}_{\ell}(\eta) f^{(b)}_t(\eta) m_{\ell} m_{t} \E\left\{\widetilde{h}_1^{(\ell)}(X_1)\widetilde{h}_1^{(t)}(X_1)\right\}=\sigma_{ab}$.

By applying the Cauchy--Schwarz inequality, \eqref{eq:type2} can be bounded by
\begin{equation}
(n-1) \sum_{j=1}^n \left\| \widehat{\eta}_{-j}-\widehat{\eta}\right\| ^2_2 \| \nabla f^{(a)}(\eta)\| _2  \left\| \nabla\left\{ f^{(b)}(\tau_j) - f^{(b)}(\eta) \right\}\right\| _2.
\label{eq:inequality1}
\end{equation}
Now, since $q$ is finite, following an argument similar to that used for term \eqref{eq:type1}, it follows that $(n-1)\sum_{j=1}^n \left\| \widehat{\eta}_{-j}-\widehat{\eta}\right\| ^2_2=O_{\Prob}(1)$. By assumption $\| \nabla f^{(a)}(\eta)\| _2$ is bounded. Finally, from  Lemma~\ref{lemma:etaj-eta} and \eqref{eq:assymptheta} it follows that
\begin{equation}
\max_{j \in \{1,\ldots,n\}}\| \tau_j-\eta\| _2 \leq  \max_{j \in \{1,\ldots,n\}} \| \widehat{\eta}_{-j}-\eta\| _2 + \| \widehat{\eta}-\eta\| _2=O_{\Prob}\left(n^{-1/2}\right).   
\label{eq:maxtau-eta}
\end{equation}
With this, given the fact that $f^{(b)}$ has continuous first partial derivatives near $\eta$, we get that
\begin{equation*}
\max_{j \in \{1,\ldots,n\}}\left\| \nabla\left\{ f^{(b)}(\tau_j) - f^{(b)}(\eta) \right\}\right\| _2=o_{\Prob}(1).
\end{equation*}
This result, combined with \eqref{eq:inequality1}, implies that the term in \eqref{eq:type2} converges in probability to 0. The same arguments as those used for term \eqref{eq:type2} apply to terms \eqref{eq:type3}--\eqref{eq:type5}, yielding convergence to zero in probability.
\end{Prooflx}

\begin{Proofpx}{Proposition:WTS}
Under $H_0$, it follows that $C\theta = 0$. Let
\begin{equation}
Z_n=\sqrt{n} \left(\widehat{\theta} - \theta \right),
\label{eq:ZN}
\end{equation}
then from \eqref{definition:sigma} we have that
\begin{equation}
C Z_n =   C \sqrt{n} (\widehat{\theta} - \theta)= \sqrt{n}C  \widehat{\theta} \tl \black{Z_1},
\label{eq:limCZ}
\end{equation}
where $\black{Z_1}\sim \mathcal{N}_{kd}\left(0, C \Sigma C^\top\right)$.
Moreover, from  Lemma~\ref{Lemma:varconsist} and the Continuous Mapping Theorem, 
it follows that $C\widehat{\Sigma} C^\top \tp C\Sigma C^\top$. 
Recalling that $C\Sigma C^\top$ has rank $(k-1)d$, and since the eigenvalue map is continuous, 
it follows by construction of the $g$-inverse that
\begin{equation}
\left(C\widehat{\Sigma} C^\top\right)^{-} \tp \left(C \Sigma C^\top\right)^{-}.
\label{eq:convginverse}
\end{equation}
Note that the statistic can be written as a quadratic form $T_n = (CZ_n)^\top (C \widehat{\Sigma} C^\top)^{-} (CZ_n)
$, with $Z_n$ defined in \eqref{eq:ZN}. Then,
from \eqref{eq:limCZ}, \eqref{eq:convginverse} and the Continuous Mapping Theorem, it follows that $T_n \tl {\black{Z_1}}^\top (C \Sigma C^\top)^{-} \black{Z_1}$. Now, by applying Theorem 4.4.3 of \citet{graybill1976linear}, we get the result.
\end{Proofpx}

\begin{Prooflx}{Lemma:consistencyWTS}
From \eqref{eq:SLLNtheta}, \eqref{eq:convginverse} and the Continuous Mapping Theorem, it follows that
\begin{equation*}
T_n/n \tp (C\theta)^{\top} \left(C \Sigma C^\top\right)^{-} (C\theta) =: \Delta_T.
\end{equation*} Now, note that for any matrix $M \ge 0$ and any vector $v$ in the column space of $M$, the quadratic form $v^{\top} M^{-} v$ is non-negative and vanishes if and only if $v = 0$. Since $C\theta$ belongs to the column space of $C\Sigma C^{\top}$ because $\Sigma>0$, it follows that $\Delta_T = 0$ if and only if $C\theta = 0$, which is precisely the null hypothesis $H_0$.
\end{Prooflx}

\begin{Proofpx}{Proposition:PowerWTS}
Let us analyze the asymptotic behavior of the test statistic
\begin{equation}
T_n = n \widehat{\theta}^{\top} C^{\top} \left(C \widehat{\Sigma} C^{\top}\right)^{-} C \widehat{\theta} = \left(\sqrt{n} C \widehat{\theta}\right)^{\top} \left(C \widehat{\Sigma} C^{\top}\right)^{-} \left(\sqrt{n} C \widehat{\theta}\right).
\label{eq:TnDecomp}
\end{equation}
From Assumption~\ref{assumption}, we have $\theta_n = \theta_0 + n^{-\gamma}b$ and, since $C\theta_0 = 0$, it follows that $C\theta_n = C(\theta_0 + n^{-\gamma}b) = n^{-\gamma}Cb$. Moreover, letting $Z_n = \sqrt{n}(\widehat{\theta} - \theta_n)$, the independence of the samples and \eqref{eq:compsamp} yield $Z_n \tl Z_2$, where $Z_2 \sim \mathcal{N}_{kd}(0, \Sigma)$, and $\widehat{\Sigma} \tp \Sigma$, so that \eqref{eq:convginverse} holds. We can rewrite the term $\sqrt{n} C \widehat{\theta}$ as $\sqrt{n} C \widehat{\theta} = C \sqrt{n}(\widehat{\theta} - \theta_n) + \sqrt{n} C \theta_n = C Z_n + n^{1/2 - \gamma} C b$. 
Substituting this last expression into \eqref{eq:TnDecomp} and recalling \eqref{eq:convginverse}, we obtain that $T_n$ behaves asymptotically as the quadratic form $T^{\prime}_n$, with
\begin{equation}
T^{\prime}_n = (C Z_n + n^{1/2 - \gamma} C b)^{\top} (C \Sigma C^{\top})^{-} (C Z_n + n^{1/2 - \gamma} C b).
\label{eq:TnAsymp}
\end{equation}
We now analyze the limit based on the value of $\gamma$. If $0\le\gamma<1/2$, the term $n^{1/2-\gamma}Cb$ dominates and tends to infinity (since $Cb\neq0$). Consequently, the quadratic form grows without bound, implying $T_n\tp\infty$, and hence $\mathcal{P}=1$. If $\gamma=1/2$, we have $n^{1/2-\gamma}=1$. From \eqref{eq:TnAsymp}, it follows that $T_n\tl(C\black{Z}_2+Cb)^\top(C\Sigma C^\top)^{-}(C\black{Z}_2+Cb)$. Let $Y=C\black{Z}_2+Cb$. From the properties of the multivariate normal distribution, $Y\sim\mathcal{N}_{kd}(Cb,C\Sigma C^\top)$. Thus, $T_n$ converges in distribution to $\chi^2_{(k-1)d}(\delta^2)$, where $\delta^2=(Cb)^\top(C\Sigma C^\top)^{-}(Cb)$. Since $Cb\neq0$ and $\Sigma>0$, we have $\delta^2>0$. As the non-central chi-squared distribution stochastically dominates the central one, the asymptotic power satisfies $\mathcal{P}=p_W(\alpha)$, with $\alpha<p_W(\alpha)<1$. Finally, if $\gamma>1/2$, then $n^{1/2-\gamma}Cb\to0$, so the limiting distribution coincides with that under $H_0$. Therefore, $\mathcal{P}=\alpha$.
\end{Proofpx}

\begin{Proofpx}{Proposition:ATS}
Following a similar reasoning to that in the proof of  Proposition~\ref{Proposition:WTS} and applying Theorem 4.4.4 of \citet{graybill1976linear}, we get the result.
\end{Proofpx}

\begin{Prooflx}{Lemma:consistencyATS}
From \eqref{eq:SLLNtheta}, and the Continuous Mapping Theorem, we have that $Q_n/n \tas \theta^{\top}C \theta=\| C\theta\| _2^2:=\Delta_Q$,
where we have used that $C$ is a symmetric and idempotent matrix. Since the norm is non-negative, $\Delta_Q \geq 0$, and $\Delta_Q = 0$ if and only if $C\theta = 0$. By the definition of $C$, this is equivalent to $\theta_1=\cdots=\theta_k$, which is precisely $H_0$.
\end{Prooflx}

\begin{Proofpx}{Proposition:PowerATS}
Note that $Q_n$ can be expressed as
\begin{equation}
Q_n= (Z_n+ \sqrt{n}\theta_n)^{\top} C (Z_n+ \sqrt{n}\theta_n),
\label{eq:decompQN}
\end{equation}
with $Z_n = \sqrt{n}(\widehat{\theta} - \theta_n)$ as in the proof of Proposition~\ref{Proposition:PowerWTS}. From Assumption~\ref{assumption}, since $C\theta_0=0$, we can write  $C (Z_n+ \sqrt{n}\theta_n)=C(Z_n + \sqrt{n}\theta_0+n^{1/2-\gamma} b)=C(Z_n+n^{1/2-\gamma}b)$.
From the previous equation, \eqref{eq:decompQN} and recalling that $C$ is symmetric and idempotent, we get that $ Q_n= (Z_n+n^{1/2-\gamma}b)^{\top} C(Z_n+n^{1/2-\gamma}b)$.
Now the reasoning is similar to that used in the proof of  Proposition~\ref{Proposition:PowerWTS} and hence omitted.
\end{Proofpx}

Before proving  Theorem~\ref{theorem:bootstrap_consistency}, we show a result concerning the centered jackknife pseudovalues.

\begin{lemma}
Let $X_1, \dots, X_n$ be a random sample and let $\widehat{\eta}=F(X_1,\ldots,X_n) \in \R$ be a U-statistic of order $m$ with kernel $h$ such that 
$\mathbb{E}\left\{| h(X_1,\dots,X_m)| ^2\right\} < \infty$, $\V\{h_1(X_1)\}=\zeta_1>0$ and 
$\mathbb{E}\left\{h(X_1,\dots,X_m)\right\} = \eta$. Denote by 
$\widehat{\eta}_{-j}=F(X_1,\ldots,X_{j-1},X_{j+1}, \ldots,X_n)$. Then
\begin{equation*}
\widehat{\eta}_{-j} - \widehat{\eta}= -\frac{m}{n-1} \left\{ h_1(X_j) - \overline{h}_1 \right\}+O_{\Prob}(n^{-3/2}),
\end{equation*}
where $\overline{h}_1=n^{-1}\sum_{\ell=1}^n h_1(X_{\ell})$.
\label{lemma:etaj-eta_nomax}
\end{lemma}

\begin{Proof}
From the Hoeffding decomposition  for a $U$-statistic, \black{(S6) in the Supplementary Material}, we have that
\begin{equation}
\widehat{\eta} = \eta + \frac{m}{n} \sum_{\ell=1}^n g_1(X_{\ell}) + R_n,
\label{eq:hoeffding_full}
\end{equation}
where $R_n$ is the degenerate part of order $O_{\Prob}(n^{-1})$. Similarly,  we can write
\begin{equation}
\widehat{\eta}_{-j} = \eta + \frac{m}{n-1} \sum_{\ell \neq j} g_1(X_{\ell}) + R_{n,-j}.
\label{eq:hoeffding_jack}
\end{equation}
Subtracting \eqref{eq:hoeffding_full} from \eqref{eq:hoeffding_jack}, we decompose the difference $\widehat{\eta}_{-j}-\widehat{\eta}$ into a linear component $D_{n,j}^{(1)}$ plus a remainder component $D_{n,j}^{(2)}$. The difference in the linear terms is given by
\begin{align*}
D_{n,j}^{(1)} &= m \left( \frac{1}{n-1} \sum_{\ell \neq j} g_1(X_{\ell}) - \frac{1}{n} \sum_{\ell=1}^n g_1(X_{\ell}) \right) 
= m \left[ \frac{1}{n-1} \left\{ \sum_{\ell=1}^n g_1(X_{\ell}) - g_1(X_j) \right\} - \frac{1}{n} \sum_{\ell=1}^n g_1(X_{\ell}) \right] \\ & =m \left[ \left( \frac{1}{n-1} - \frac{1}{n} \right) \sum_{\ell=1}^n g_1(X_{\ell}) - \frac{1}{n-1} g_1(X_j) \right] = \frac{m}{n-1} \left\{ \frac{1}{n} \sum_{\ell=1}^n g_1(X_{\ell}) - g_1(X_j) \right\} \\
&= -\frac{m}{n-1} \left\{ g_1(X_j) - \overline{g}_1 \right\} = -\frac{m}{n-1} \left\{ h_1(X_j) - \overline{h}_1 \right\},
\end{align*}
where $\overline{g}_1=n^{-1}\sum_{\ell=1}^n g_1(X_{\ell})$ and $g_1(X_j)=h_1(X_j)-\eta$. Since $g_1(X_j) = O_{\Prob}(1)$ and $\overline{g}_1 = O_{\Prob}(n^{-1/2})$, we have that $D_{n,j}^{(1)}=O_{\Prob}(n^{-1})$. We now bound $D_{n,j}^{(2)} = R_{n,-j} - R_n$. It suffices to consider the second-order degenerate kernel $g_2$, as higher-order terms decay faster. Let $S_n = \sum_{k<\ell} g_2(X_k, X_{\ell})$. Then, up to constant factors, $R_n = \binom{n}{2}^{-1} S_n$ and $R_{n,-j}$ is explicitly given by removing the terms involving $X_j$, as follows
\begin{equation*}
R_{n,-j} = \binom{n-1}{2}^{-1} \left\{ S_n - W_j \right \},
\end{equation*}
where $W_j = \sum_{k \neq j} g_2(X_k, X_j)$. 
Then, we can write
\begin{equation*}
D_{n,j}^{(2)} = \left\{ \binom{n-1}{2}^{-1} - \binom{n}{2}^{-1} \right\} S_n - \binom{n-1}{2}^{-1} W_j
= \frac{4}{n(n-1)(n-2)} S_n - \frac{2}{(n-1)(n-2)} W_j.
\end{equation*}
We analyze the order of these two terms separately. First,  since $R_n=O_{\Prob}(n^{-1})$ it follows that $S_n=O_{\Prob}(n)$. Multiplying by the coefficient $O(n^{-3})$, the first term is $O_{\Prob}(n^{-2})$. Second, for the term involving $W_j$, we determine its asymptotic order by computing its variance, it follows that $\V(W_j) = \E\{\V(W_{j} \vert X_{j})\} + \V\{\E(W_{j} \vert X_{j})\}$. We analyze the two components on the right-hand side of the previous equality:

\noindent\textit{Expectation of conditional variance.} Conditional on $X_j$, the terms in the sum $W_j$ are independent and centered. Thus,
\begin{equation*}
\V(W_{j} \vert X_{j}) = \sum_{k \neq j} \V\{g_2(X_k, X_j)| X_j\} = (n-1) \E\{g_2^2(X_k, X_j)| X_j\}.
\end{equation*}
Taking the expectation over $X_j$, we obtain $\E\{\V(W_{j} \vert X_{j})\} = (n-1) \E\{g_2^2(X_1, X_2)\} = O(n)$.

\noindent\textit{Variance of conditional expectation}. Due to the degeneracy property of the canonical kernel $g_2$, see  \black{(S5) in the Supplementary Material}, the conditional expectation vanishes, leading to
$\E(W_{j} \vert X_{j}) = \sum_{k \neq j} \E\{g_2(X_k, X_j)\vert X_j\} = 0.
$
Therefore, $\V\{\E(W_{j} \vert X_{j})\} = \V(0) = 0$.

Combining these results, the total unconditional variance is $\V(W_j) = O(n)$. From Markov's inequality, $W_j = O_{\Prob}(n^{1/2})$. Finally, substituting back into the expression for $D_{n,j}^{(2)}$, the contribution of the second term is
$O_{\Prob}(n^{-3/2})$.
\end{Proof}

\begin{Proofx}{theorem:bootstrap_consistency}
Let $Z_n^{*} = \sqrt{n}(\widehat{\theta}^{*} - \widehat{\theta})$. We aim to show that, conditional on the data $\mathcal{X}$, $Z_n^{*} \tl \mathcal{N}_{kd}(0, \Sigma)$ in probability. From the Cram\'er--Wold Theorem (see Section 1.5.2 in \citet{serfling1980approximation}), it suffices to show that for any arbitrary fixed vector $a \in \mathbb{R}^{kd}$, the linear combination $a^\top Z_n^{*}$ converges in distribution to $\mathcal{N}(0, \sigma_a^2)$, where $\sigma_a^2 = a^\top \Sigma a$. Note that $\sigma_a^2 > 0$ provided that $a \neq 0$ and $\Sigma$ is positive definite.

Consider the scalar random variable derived from the projection, $S_n^{*} := a^\top Z_n^{*} = \sum_{i=1}^k a_i^\top Z_{n,i}^{*}$,
where $a = \left[a_1^\top, \dots, a_k^\top\right]^\top$ with $a_i \in \mathbb{R}^d$. Substituting the bootstrap estimator definition into $Z_{n,i}^{*}$, we can express $S_n^{*}$ as a sum of independent random variables conditional on $\mathcal{X}$ as
$
S_n^{*} = \sum_{i=1}^k \sum_{j=1}^{n_i} Y_{ij}^{*},
$ 
where $Y_{ij}^{*} = c_{ij} W_{ij}$, and the scalar coefficients are defined as
\begin{equation}
c_{ij} = \sqrt{\frac{n}{n_i}} \frac{1}{\sqrt{n_i}} a_i^\top (T_{i,j} - \overline{T}_i).
\label{eq:cij}
\end{equation}
Conditional on $\mathcal{X}$, the variables $\{Y_{ij}^{*}\}$ are independent because the bootstrap weights $\{W_{ij}\}$ are i.i.d. with $\mathbb{E}(W_{ij})=0$ and $\V(W_{ij})=1$. Consequently, $\mathbb{E}_*(Y_{ij}^{*}) = 0$ and  $s_n^2 := \V_*(S_n^{*}) = \sum_{i=1}^k \sum_{j=1}^{n_i} c_{ij}^2 = a^\top \widehat{\Sigma} a$.
From  Lemma~\ref{Lemma:varconsist}, we know that $\widehat{\Sigma} \tp \Sigma$, and thus, $s_n^2 \tp \sigma_a^2$.

To establish the asymptotic normality we use the Lindeberg--Feller Central Limit Theorem. We must verify the Lindeberg condition
\begin{equation*}
\ell_n(\varepsilon) = \frac{1}{s_n^2} \sum_{i=1}^k \sum_{j=1}^{n_i} \mathbb{E}_*\left[ c_{ij}^2 W_{ij}^2 \mathbb{I}\left\{ c_{ij}^2 W_{ij}^2 > \varepsilon s_n^2 \right\} \right] \tp 0, \quad \forall \varepsilon>0,
\end{equation*}
where $ \mathbb{I}(\cdot) $ denotes the indicator function.

\noindent\textit{Step 1: Asymptotic approximation of coefficients $c_{ij}$.} We approximate the pseudo-values using a first-order Taylor expansion of the vector-valued function $f: \mathbb{R}^q \to \mathbb{R}^d$. By definition, $\widehat{\theta}_{i,(-j)} = f(\widehat{\eta}_{i,(-j)})$. Expanding around $\widehat{\eta}_i$, we obtain
$\widehat{\theta}_{i,(-j)} = f(\widehat{\eta}_i) + \nabla f(\widehat{\eta}_i)^\top \left( \widehat{\eta}_{i,(-j)} - \widehat{\eta}_i \right) + O_{\Prob}\left( \left\|  \widehat{\eta}_{i,(-j)} - \widehat{\eta}_i \right\| ^2 \right)$, where $\nabla f(\widehat{\eta}_i)$ denotes the $q \times d$ matrix of component-wise gradients, i.e.,~$\nabla f = [\nabla f_1, \dots, \nabla f_d]$. Now from the last equation and  Lemma~\ref{lemma:etaj-eta_nomax}, we get
\begin{equation*}
T_{i,j} = n_i \widehat{\theta}_i - (n_i-1) \widehat{\theta}_{i,(-j)} = f(\widehat{\eta}_i) + m \nabla f(\widehat{\eta}_i)^\top \left\{ h_1(X_{ij}) - \overline{h}_{1,i} \right\} + O_{\Prob}(n_i^{-1/2}).
\end{equation*}
Averaging over $j$, the linear term vanishes, so $\overline{T}_i = f(\widehat{\eta}_i) + O_{\Prob}(n_i^{-1/2})$. Therefore, the centered pseudo-values are dominated by the linear influence function $T_{i,j} - \overline{T}_i = M_h \nabla f(\widehat{\eta}_i)^\top \left\{ h_1(X_{ij}) - \overline{h}_{1,i} \right\} + o_{\Prob}(1)$.

By substituting the previous expression into \eqref{eq:cij}, we obtain
\begin{equation}
c_{ij}^2 = \frac{n
}{n_i^2} \left[ a_i^\top \nabla f(\widehat{\eta}_i)^\top M_h \left\{ h_1(X_{ij}) - \overline{h}_{1,i} \right\} \right]^2 + o_{\Prob}(n^{-1}).
\label{eq:cij2}
\end{equation}
\textit{Step 2: Vanishing maximum coefficient.
} We first establish that the projection $h_1$ inherits the finite second moment from the kernel $h$. Applying Jensen's inequality componentwise for the conditional expectations with the convex function $\phi(t) = t^2$, we obtain $\left\lvert  h_1^{(r)}(X_1)\right\rvert  ^2 = |  \mathbb{E}\left\{h^{(r)}(X_1, \dots, X_m) \vert  X_1\right\} | ^2 \le \mathbb{E}\left\{ | h^{(r)}(X_1, \dots, X_m)\vert ^2 |  X_1 \right\}$. Taking expectations yields $\mathbb{E}\left\{ | h_1^{(r)}(X_1)| ^2 \right\} \le \mathbb{E}\{ | h^{(r)}(X_1, \dots, X_m)| ^2 \} < \infty$, $r\in \{1,\ldots,q\}$.

A known probabilistic result states that for i.i.d. random vectors $U_1, \dots, U_n$ with $\mathbb{E}\| U\| ^2 < \infty$, the maximum satisfies $\max_{i \in \{1,\ldots, n\}} \| U_i\|  = o_{\Prob}\left(n^{1/2}\right)$ (see, e.g.,~Lemma 11.2 in \citet{owen2001empirical}). Consequently, applying this to the centered projection \black{(S2) in the Supplementary Material}, $\max_{i,j} \left\|  \widetilde{h}_1(X_{ij}) \right\| ^2 = o_{\Prob}(n)$. To handle the centered term, we use the inequality $\| a-b\| ^2 \le 2\| a\| ^2 + 2\| b\| ^2$, which implies that
\begin{equation*}
\max_{i,j} \left\|  h_1(X_{ij}) - \overline{h}_{1,i} \right\| ^2 \le 2 \max_{i,j} \left\|  \widetilde{h}_1(X_{ij}) \right\| ^2 + 2 \max_{i} \left\|  \overline{h}_{1,i}-\eta \right\| ^2.
\end{equation*} Since $\left\| \overline{h}_{1,i}-\eta_i \right\|  = O_{\Prob}(n^{-1/2})$ from the CLT applied to $h_{1}(X_{i1}),\ldots, h_{1}(X_{in_i})$, and assuming the number of groups $k$ is fixed, the second term is $O_{\Prob}(n^{-1})$. Thus, the maximum is dominated by the individual observations, yielding $\max_{i,j} \|  h_1(X_{ij}) - \overline{h}_{1,i} \| ^2 = o_{\Prob}(n)$. Substituting this back into \eqref{eq:cij2} and recalling \eqref{eq:compsamp}, we have that $\max_{i,j} c_{ij}^2 = O(n^{-1})  o_{\Prob}(n) = o_{\Prob}(1)$.

\noindent\textit{Step 3: Verifying Lindeberg's Condition.} We bound the indicator function by using the maximum coefficient. Specifically, if $c_{ij}^2 W_{ij}^2 > \varepsilon s_n^2$, then it must hold that $(\max_{k,\ell} c_{k\ell}^2) W_{ij}^2 > \varepsilon s_n^2$.
Defining the threshold $M_n = \varepsilon s_n^2/(\max_{i,j} c_{ij}^2)$, we have the inequality
$ \mathbb{I}\left\{ c_{ij}^2 W_{ij}^2 > \varepsilon s_n^2 \right\} \le \mathbb{I}\left\{ W_{ij}^2 > M_n \right\}$. Substituting this into the expression of $\ell_n(\varepsilon)$
\begin{equation*}
\ell_n(\varepsilon) = \frac{1}{s_n^2} \sum_{i,j} c_{ij}^2 \mathbb{E}_*\left[ W_{ij}^2 \mathbb{I}\left\{ c_{ij}^2 W_{ij}^2 > \varepsilon s_n^2 \right\} \right] 
\le \frac{1}{s_n^2} \sum_{i,j} c_{ij}^2 \mathbb{E}_{*}\left[ W_{ij}^2 \mathbb{I}\left\{ W_{ij}^2 > M_n \right\} \right].
\end{equation*}
Since the weights are i.i.d., the expectation term $\mathbb{E}_*[ W_{ij}^2 \mathbb{I}\{ W_{ij}^2 > M_n \} ]$ is common to all terms in the sum and can be factored out, yielding $\ell_n(\varepsilon) \le \mathbb{E}_*\left\{ W^2_{11} \,\mathbb{I}\left( W_{11}^2 > M_n \right) \right\}$. As established in Step 2, $\max_{i,j} c_{ij}^2 = o_{\Prob}(1)$ and $s_n^2 \tp \sigma_a^2 > 0$, which implies that  $M_n \tp \infty$.
From \eqref{eq:random_weights}, we have that $\mathbb{E}_*(W_{11}^2)=\mathbb{E}(W_{11}^2)=1$. The Dominated Convergence Theorem ensures that $\mathbb{E}_*\left\{ W_{11}^2 \mathbb{I}\left( W_{11}^2 > M_n \right) \right\} \tp  0$. Thus, $\ell_n(\varepsilon) \tp 0$ for any $\varepsilon > 0$.

The Lindeberg--Feller CLT implies that, conditional on $\mathcal{X}$, the scalar projection satisfies $S_n^{*} = a^\top Z_n^{*} \tl \mathcal{N}(0, \sigma_a^2)$, in probability. Now, since $a \in \mathbb{R}^{kd}$ was arbitrary, the Cram\'er--Wold device allows us to conclude that  $Z_n^{*} \tl \mathcal{N}_{kd}(0, \Sigma)$, in probability.

Finally, the convergence of the quadratic form $Q_n^{*}$ follows from the Continuous Mapping Theorem. Since the limiting distribution is continuous, Pólya's Theorem (see Section 1.5.3 in \citet{serfling1980approximation}) ensures that the result holds.
\end{Proofx}

In order to prove  Theorem~\ref{theorem:assymptoticdinfty}, we first state the following lemma.

\begin{lemma}\label{lemma:eigenconsistency}
Let $C\in\mathbb{R}^{kd\times kd}$ be an orthogonal projection matrix, i.e.,~$C=C^\top=C^2$, let $\Sigma\in\mathbb{R}^{kd\times kd}$ be defined as in~\eqref{definition:sigma}, and let $\Sigma_n=K_n\Sigma$, where $K_n=\operatorname{diag}\{\kappa_1(n/n_1)I_d,\ldots,\kappa_k(n/n_k)I_d\}$. Assume that \eqref{eq:compsamp} and \eqref{eigencontrol} hold. Then,
\begin{equation*}
\max_{1\le j\le kd}\big| \lambda_j(C\Sigma_n C)-\lambda_j(C\Sigma C)\big| \longrightarrow0.
\end{equation*}
Moreover, $M_1d\le\tr\{(C\Sigma_n C)^2\}\le M_2d$.
\end{lemma}

\begin{Proof}
By construction, $\Sigma_n-\Sigma=(K_n-I_{kd})\Sigma$ and $C(\Sigma_n-\Sigma)C=C(K_n-I_{kd})\Sigma C$. Recall that for any conformable matrices $A$ and $B$, the spectral norm satisfies the submultiplicative property (see, e.g.,~Result 4.65 in \citeauthor{seber2008matrix}, \citeyear{seber2008matrix})
\begin{equation}
\lVert AB\rVert\le\lVert A\rVert\lVert B\rVert.
\label{eq:submult}
\end{equation}
Since $K_n-I_{kd}$ and $\Sigma$ are block-diagonal, and the diagonal blocks of $K_n-I_{kd}$ are multiples of $I_d$, we have that $(K_n-I_{kd})\Sigma=\Sigma^{1/2}(K_n-I_{kd})\Sigma^{1/2}$. By using \eqref{eq:submult}, we obtain $\lVert C(\Sigma_n-\Sigma)C\rVert\le\lVert C\Sigma^{1/2}\rVert\lVert K_n-I_{kd}\rVert\lVert \Sigma^{1/2}C\rVert=\lVert K_n-I_{kd}\rVert\lambda_1(C\Sigma C)$, where we have used that $\lVert \Sigma^{1/2}C\rVert^2=\lVert C\Sigma C\rVert=\lambda_1(C\Sigma C)$. Since $\lambda_1(C\Sigma C)$ is bounded from \eqref{eigencontrol}, and $\lVert K_n-I_{kd}\rVert=\max_i| \kappa_i n/n_i-1| \to0$ from~\eqref{eq:compsamp}, we obtain $\lVert C(\Sigma_n-\Sigma)C\rVert\to0$. Applying Weyl's Perturbation Theorem (see, e.g.,~Corollary~III.2.6 in~\citeauthor{bhatia2013matrix}, \citeyear{bhatia2013matrix}) yields the first claim.

Since $C\Sigma_n C$ is symmetric and positive semidefinite, $\lambda_j\{(C\Sigma_n C)^2\}=\lambda_j^2(C\Sigma_n C)$, $1\le j\le kd$. From the first part, $\lambda_j(C\Sigma_n C)\to\lambda_j(C\Sigma C)$ uniformly in~$j$. Hence, $\lambda_{\min}^{+}(C\Sigma_n C)\ge\lambda_{\min}^{+}(C\Sigma C)+o(1)$ and $\lambda_1(C\Sigma_n C)\le\lambda_1(C\Sigma C)+o(1)$. From \eqref{eigencontrol}, we obtain $kd(M_1+o(1))^2\le\tr\{(C\Sigma_n C)^2\}=\sum_{j=1}^{kd}\lambda_j^2(C\Sigma_n C)\le kd(M_2+o(1))^2$, which proves the result.
\end{Proof}

The next lemma is the heterogeneously distributed version of Lemma 1 of \citet{PENG201822}.
\begin{lemma}\label{lem:4thmoment-trace}
Let $\{X_s\}_{s=1}^{n}\subset\mathbb{R}^{kd}$ be independent, centered random vectors with
$\Sigma_s=\V(X_s)$ and $\E\| X_s\| _2^4<\infty$, $\forall s$. Define $S_k=\sum_{s=1}^{k}X_s$.
Then,
\begin{equation}
\V\left(\| S_k\| _2^2\right) 
=  \sum_{s=1}^k \V\left(\| X_s\| _2^2\right) + 4 \sum_{1\leq r < s \leq k} \operatorname{tr}\left(\Sigma_s \Sigma_r\right).
\label{eq:vars2}
\end{equation}
In addition, $\V\left(\sum_{k=1}^m \| S_k\| _2^2 \right) \leq  m^2 \sum_{s=1}^m \V\left(\| X_s\| _2^2\right) +4m^2  \sum_{1\leq r < s \leq m} \operatorname{tr}\left(\Sigma_s \Sigma_r\right)$.
\end{lemma}
\begin{Proof}
The proof follows an analogous reasoning to that of the iid case. First, note that
\begin{equation}
\E\left(\| X_s\| ^2_2\right) 
= \E\left(X_s^{\top} X_s \right)
= \E\left\{\tr\left(X_s \, X_s^{\top}\right)\right\} 
= \tr \left\{\E\left(X_s \, X_s^{\top}\right)\right\} 
= \tr\left(\Sigma_s\right)
\label{eq:EX2}
\end{equation}
and
\begin{equation}
\E\left\{(X_s^\top X_r)^2\right\}
= \E\left\{\operatorname{tr}\left(X_s^\top X_r X_r^\top X_s\right)\right\} 
= \E\left\{\operatorname{tr}\left(X_s X_s^\top X_r X_r^\top\right)\right\}
= \operatorname{tr}\left\{\E\left(X_s X_s^\top X_r X_r^\top\right)\right\} 
= \operatorname{tr}\left(\Sigma_s \Sigma_r\right),
\label{eq:EXrXs}
\end{equation}
where we have applied the cyclic property of the trace. Then,
\begin{equation}
\E\left(\| S_k\| ^2_2\right)=\sum_{r=1}^k\sum_{s=1}^k \E(X_r^{\top} X_s)=\sum_{s=1}^k \E\| X_s\| ^2_2=\sum_{s=1}^k \tr\left(\Sigma_s\right), 
\label{eq:ES2}
\end{equation}
where we have applied \eqref{eq:EX2}. Similarly,
\begin{align}
\E\left(\| S_k\| ^4_2\right)
&= \sum_{r=1}^k \sum_{s=1}^k \sum_{\ell=1}^k \sum_{t=1}^k 
\E(X_r^{\top} X_s \, X_{\ell}^{\top} X_t) \nonumber\\
&= \sum_{s=1}^k \E\left(\| X_s\| _2^4\right)
+ 2 \sum_{1 \le r < s \le k} \E(\| X_r\| _2^2)\, \E(\| X_s\| _2^2) + 4 \sum_{1 \le r < s \le k} \E \{(X_s^\top X_r)^2\} \nonumber\\
&= \sum_{s=1}^k \E\left(\| X_s\| _2^4\right)
+ 2 \sum_{1 \le r < s \le k} \operatorname{tr}(\Sigma_r)\operatorname{tr}(\Sigma_s) + 4 \sum_{1 \le r < s \le k} \operatorname{tr}(\Sigma_s \Sigma_r),
\label{eq:ES4}
\end{align}
where we have applied \eqref{eq:EX2} and \eqref{eq:EXrXs}. With this, from \eqref{eq:ES2} and \eqref{eq:ES4}, we get that
\begin{align*}
\V\!\left(\| S_k\| _2^2\right)
&= \E\left(\| S_k\| _2^4\right)
- \left\{\E\left(\| S_k\| _2^2\right)\right\}^2 
= \sum_{s=1}^k \E\left(\| X_s\| _2^4\right)
+ 4 \sum_{1 \le r < s \le k} \operatorname{tr}(\Sigma_s \Sigma_r)
- \sum_{s=1}^k \left\{\tr(\Sigma_s)\right\}^2 \\
&= \sum_{s=1}^k \E\left(\| X_s\| _2^4\right)
- \sum_{s=1}^k \left\{\E\left(\| X_s\| _2^2\right)\right\}^2
+ 4 \sum_{1 \le r < s \le k} \operatorname{tr}(\Sigma_s \Sigma_r),
\end{align*}
which proves \eqref{eq:vars2}. Finally, note that the covariance of $\| S_r\| ^2_2$ and $\| S_s\| ^2_2$ is the variance of $\| S_{\min\{r,s\}}\| ^2_2$. Therefore,
\begin{equation*}
\V\left(\sum_{k=1}^m \| S_k\| _2^2 \right)=\sum_{k=1}^m \V\left(\| S_k\| _2^2 \right)+2 \sum_{1 \le r < s \le m}\C\left(\| S_r\| ^2_2,\| S_s\| ^2_2\right)=\sum_{k=1}^m \V\left(\| S_k\| _2^2 \right) \left\{1+2(m-k)\right\}.
\end{equation*}
Applying \eqref{eq:vars2}, one gets the desired result.
\end{Proof}

\begin{Proofx}{theorem:assymptoticdinfty}
We divide the proof into several steps.

\noindent\textit{Step 1: Quadratic form decomposition.}
Following a similar reasoning to that in the proof of Theorem 1 in Section 3.2.1 of \citet{Lee1990ustatistics}, we can express
\begin{equation*}
\sqrt{n}(\widehat{\eta}_i-\eta_i)=\sqrt{\frac{n}{n_i}}\frac{1}{\sqrt{n_i}}M_h\sum_{j=1}^{n_i}\widetilde{h}_1 (X_{ij}) +  R^H_{i}:= L^{H}_i + R^{H}_i, \quad i \in \{1,\ldots,k\},
\end{equation*}
where $M_h$ is defined in \eqref{eq:Mh}. From \eqref{ass:finitenessh2_bounzeta}, \eqref{eq:compsamp} we have that  $L^{H}_i=O_{\Prob}(1)$ and $R^{H}_i=O_{\Prob}(n^{-1/2})$. In addition, $\E(L^{H}_{i})=0$ and $\V(L^H_{i})=(n/n_i) V_i$, where $V_i$ is defined in \eqref{eq:definitionvi}.
With this, applying a Taylor expansion of $f$ at $\eta_i$,
\begin{equation*}
\sqrt{n}\,(\widehat{\theta}_i-\theta_i)
= D_i\,(L^H_i+R^H_i) + R^T_i= L_i +R_i, \; \text{ with } L_i=D_i L_i^H \, \text{ and } \, R_i=D_i R_i^H + R_i^T, \quad i \in \{1,\ldots,k\}.
\end{equation*}
Notice that each component of $R^H_i$ is $O_{\Prob}(n^{-1/2})$. The fact that $(\widehat{\eta}_i-\eta_i)=O_{\Prob}(n^{-1/2})$ componentwise and that $f$ has continuous second partial derivatives in a neighborhood of $\eta_i$, gives that $R_i^T=O_{\Prob}(n^{-1/2})$.
Stacking the blocks $\sqrt{n}(\widehat{\theta}_1-\theta_1), \ldots, \sqrt{n}(\widehat{\theta}_k-\theta_k)$, we can write
$ Z_n=\sqrt{n}(\widehat{\theta}-\theta)=L+R,
$
where $L=\left[L_1^{\top},\ldots,L_{k}^{\top}\right]^{\top} \in \R^{kd}$, with
$\E(L)=0$ and 
$\V(L)=\operatorname{diag}\left(\frac{n}{n_1}\Sigma_1, \ldots, \frac{n}{n_k}\Sigma_k\right)= K_n  \Sigma=\Sigma_n$, where $\Sigma$ is defined in \eqref{definition:sigma}, $ K_n$  is as defined in  Lemma~\ref{lemma:eigenconsistency} and $R=\left[R_1^{\top},\ldots,R_k^{\top}\right]^{\top}\in \R^{kd}$, with
\begin{equation}
\| R\| ^2_2=O_{\Prob}(d/n).
\label{eq:orderR}
\end{equation}
Then, under $H_0$ we have that $Q_n=n \widehat{\theta}^{\top} C \widehat{\theta}= Z_n^{\top} C Z_n =  \| C L\| _2^2 + \| C R\| _2^2 + 2 (CL)^{\top} (C R)$.

\noindent\textit{Step 2: Vanishing terms.} 
We can express the statistic as
\begin{equation}
\frac{Q_n-\mu_n}{\sigma_n}=\frac{\| C L\| _2^2-\mu_n}{\sigma_n}+\frac{\| C R\| _2^2}{\sigma_n} +2 \frac{(CL)^{\top} (C R)} {\sigma_n}.
\label{eq:statinfdimQn}
\end{equation}
Note that we can write
\begin{align}
&\mu_n = \operatorname{tr}\!\bigl(C \Sigma_n\bigr)
= \operatorname{tr}\!\bigl( C \Sigma_n C\bigr)
= \E\!\left(\|  C L\| _2^2\right),
\label{eq:mun}\\
&     \sigma_n^2 = 2\,\operatorname{tr}\!\bigl\{( \Sigma_n C)^2\bigr\}
= 2\,\operatorname{tr}\!\bigl\{(C\Sigma_n C)^2\bigr\}
= 2\,\tr\!\left\{\V^2\!\left(CL\right)\right\}.
\label{eq:sigman}
\end{align}
Now, by applying \black{ Lemma~\ref{lemma:eigenconsistency}} we get
\begin{equation}
M_1 \, d \leq \sigma^2_n \leq M_2 \, d.
\label{eq:ordersigma}
\end{equation}
Recall that for any conformable matrix $A$ and vector $x$
(see, e.g.,~Result 4.73 (b) in \citeauthor{seber2008matrix}, \citeyear{seber2008matrix}),
\begin{equation}
\| Ax\| _2 \le \| A\|  \| x\| _2.
\label{eq:normcompat}
\end{equation}
By using \eqref{eq:orderR}, \eqref{eq:ordersigma} and \eqref{eq:normcompat}, we can write
\begin{equation}
\frac{\| C R\| _2^2}{\sigma_n}\leq M_1\frac{\| C\| ^2 \| R\| _2^2}{\sqrt{d}}=M_1\| R\| _2^2/\sqrt{d}=O_{\Prob}\left(\sqrt{d}/n\right) \to 0,
\label{eq:hr2to0}
\end{equation}
where we have used that $\| C\| =1$ because $C$ is idempotent and  $d/n \to 0$. For the cross-term, using the Cauchy--Schwarz inequality, from \eqref{eq:orderR} and \eqref{eq:ordersigma} we have that
\begin{equation}
\frac{\big| (CL)^{\!\top} (CR)\big| }{\sigma_n}
\;\le\; M_1 \frac{\| CL\| _2\,\| R\| _2}{\sqrt{d}}
\;=\;
O_{\Prob}\!\left\{\sqrt{\E\!\left(\| C L\| _2^2\right)}\right\}
O_{\Prob}\!\left(\sqrt{d/n}\right) O\!\left(1/\sqrt{d}\right)
= O_{\Prob}\!\left(\sqrt{d/n}\right) \to 0,
\label{eq:hlhrto0}
\end{equation}
where we used \eqref{eq:mun}, 
together with Markov's inequality and that $d/n \to 0$.  
The order  $\sqrt{\operatorname{tr}(C \Sigma_n C)}=O\left(\sqrt{d}\right)$ follows from \eqref{eigencontrol} and the definition of trace.  Finally, from \eqref{eq:statinfdimQn}, \eqref{eq:hr2to0} and \eqref{eq:hlhrto0},
\begin{equation*}
\frac{Q_n-\mu_n}{\sigma_n}=\frac{\| C L\| _2^2-\mu_n}{\sigma_n}+o_{\Prob}(1).
\end{equation*}

\noindent\textit{Step 3: Structure of the linear part.} Considering the previous equation, it remains to prove that 
$(\|  C L\| _2^2-\mu_n)/\sigma_n \tl \mathcal{N}(0,1)$.
For each population $i \in \{1,\ldots,k\}$ and each observation $j \in \{1,\ldots,n_i\}$, we define
\begin{equation*}
\xi_{n,(i,j)}
=
\Bigl[
0,\ldots,0,\;
\underbrace{\bigl\{\tfrac{n}{\,n_i\,}\,D_i\,M_h\,\widetilde{h}_1(X_{ij})\bigr\}^{\!\top}}_{\in\mathbb{R}^d\ \text{(block $i$)}},\;
0,\ldots,0
\Bigr]^{\!\top}
\in\mathbb{R}^{kd},
\qquad
\xi^C_{n,(i,j)} = C\,\xi_{n,(i,j)}.
\end{equation*}
Then
\begin{equation}
L = \frac{1}{\sqrt n}\sum_{i=1}^k\sum_{j=1}^{n_i}\xi_{n,(i,j)},
\quad
C  L = \frac{1}{\sqrt n}\sum_{i=1}^k\sum_{j=1}^{n_i}\xi^C_{n,(i,j)}.
\label{eq:LHLxi}
\end{equation}
The random vectors $\{\xi_{n,(i,j)}:i \in \{1,\ldots,k\},\ j \in \{1,\ldots,n_i\}\}$ are independent; they are not identically distributed in general (only within the same group $i$). Let  $\Xi_{(i,j)} = \V\left(\xi_{n,(i,j)}\right)$, then recalling the definition of $\Sigma_i$ in \eqref{eq:definitionsigmai}, we can write
\begin{equation*}
\Xi_{(i,j)}
= \left(\frac{n}{n_i}\right)^{\!2}
\operatorname{diag}\left(
0_{d}, \dots, 0_{d}, 
\underbrace{\Sigma_i}_{\text{$i$th position}}, 
0_{d}, \dots, 0_{d}
\right),
\qquad i \in \{1,\ldots,k\},\;\; j \in \{1,\ldots,n_i\},
\end{equation*} 
where $0_{d} \in \R^{d\times d}$ denotes the matrix of zeros.
Since $\Xi_{(i,j)}$ is block-diagonal, its eigenvalues are the union of the eigenvalues of its diagonal blocks, consisting of $(k-1)d$ zeros and the eigenvalues of $\Sigma_i$ scaled by $(n/n_i)^2$.  
Moreover, since the matrix $\Sigma$ defined in \eqref{definition:sigma} is also block-diagonal, an analogous argument shows that, under the upper bound in \eqref{eigencontrol} and \eqref{eq:compsamp}, we have $\lambda_{1}(\Xi_{(i,j)}) \le M$. Likewise, let $\Xi^C_{(i,j)} = \V\left(\xi^C_{n,(i,j)}\right)= C \Xi_{(i,j)} C$. Using the submultiplicative property \eqref{eq:submult} and the fact that
\begin{equation}
\| A\| =\lambda_1(A)
\label{eq:operatorPSD}
\end{equation}
for every symmetric matrix $A\geq 0$ (see, e.g.,~Result 4.66(b) in \citeauthor{seber2008matrix}, \citeyear{seber2008matrix}), it follows that $\lambda_{1}\left(\Xi^{C}_{(i,j)}\right) \le M$, $i \in \{1,\ldots,k\}$, $j \in \{1,\ldots,n_i\}$.
Equivalently, we can reindex $(i,j)\mapsto \ell\in\{1,\dots,n\}$ and write $\xi^C_{n,\ell}$ and $\Xi^C_{\ell}=\V\left(\xi^C_{n,\ell}\right)$, $\ell \in \{1,\ldots,n\}$, and then
\begin{equation}
\lambda_{1}\left(\Xi^{C}_{\ell}\right) \le M, \qquad \ell \in \{1,\ldots,n\}.
\label{eq:eigencontrolXi}
\end{equation}
With this notation we have that $     C \Sigma_n  C=\V( C L)=n^{-1}\sum_{\ell=1}^n \Xi^C_{\ell}$,
and then, the quantities $\mu_n$ and $\sigma_n$ in \eqref{eq:mun} and \eqref{eq:sigman} respectively, can be rewritten as
\begin{equation}
\mu_n=\operatorname{tr}\!\left(\frac{1}{n}\sum_{\ell=1}^n \Xi^C_{\ell}\right), \quad \sigma_n^2=2\,\operatorname{tr}\!\left\{\left(\frac{1}{n}\sum_{\ell=1}^n \Xi^C_{\ell}\right)^2\right\}.
\label{eq:munsigmanxi}
\end{equation}
Finally, from \eqref{eq:LHLxi} we can write $(\| C L\| _2^2-\mu_n)/{\sigma_n}= Q_{n,1} + Q_{n,2}$,
where
\begin{equation*}
Q_{n,1} = \frac{1}{n}\frac{\sum_{\ell=1}^n\left\{ \| \xi^C_{n,\ell}\| _2^2-\tr(\Xi^C_{\ell})\right\}}{\sigma_n},
\qquad
Q_{n,2} = \frac{2}{n}\sum_{1\le s<t\le n}\frac{\left(\xi^C_{n,s}\right)^\top\xi^C_{n,t}}{\sigma_n}.
\end{equation*}

\noindent\textit{Step 4: Showing that $Q_{n,1} \tp 0$ and applying martingale CLT to $Q_{n,2}$.}
Note that from the independence of $\xi_{n,1},\ldots,\xi_{n,n}$, and \eqref{eq:ordersigma}, we have that
\begin{equation}
\E\left(Q^2_{n,1}\right)= \frac{1}{n^2 \sigma_n^2} \sum_{j=1}^n \V(\| \xi^C_{n,j}\| _2^2) \leq \frac{M}{n^2 d} \sum_{j=1}^n \E(\| \xi^C_{n,j}\| _2^4).
\label{eq:Qn2}
\end{equation}
Now, recalling \eqref{eq:normcompat} and that $\|  C\| =1$, we have that
$\E(\| \xi^C_{n,j}\| _2^4)\leq \E(\| \xi_{n,j}\| _2^4)  \leq d \, \sum_{r=1}^d   \E \left\lvert  \xi_{n,j}^{(r)}\right\rvert  ^4$, $ j \in\{1,\ldots,n\}$, 
where the last inequality follows from the Cauchy--Schwarz inequality for vectors in $\mathbb{R}^d$. At this point, note that \eqref{eq:compsamp} and \eqref{eq:boundh4} imply that
$\max_{r \in \{1,\ldots,kd\}} \E | \xi^{(r)}_{n,j}| ^4 < M$, $  j \in \{1,\ldots,n\}$,
and thus,
\begin{equation}
\E(\| \xi^C_{n,j}\| _2^4)\leq Md^2, \quad j \in \{1,\ldots,n\}.
\label{eq:orderxi4}
\end{equation}
From \eqref{eq:orderxi4} and \eqref{eq:Qn2} it follows that $\mathbb{E}\bigl(Q_{n,1}^2\bigr)=O\!\left(d/n\right)\to 0$ and
hence, by Markov's inequality, $Q_{n,1}\tp0$.

To finish the proof it suffices to establish the asymptotic normality of $Q_{n,2}$. Defining
\begin{equation*}
W_0=0,\qquad
W_{\ell}=\frac{2}{n\sigma_n}\sum_{j=1}^{\ell}\xi^C_{n,j},\quad \ell \in \{1,\ldots,n\},
\end{equation*}
we can write $Q_{n,2}=\sum_{\ell=1}^n W_{\ell-1}^\top\,\xi^C_{n,\ell}$. Since we have expressed $Q_{n,2}$ as a martingale, we apply the martingale central limit theorem (see part (a) of Theorem 2.5 of \citeauthor{Helland1982}, \citeyear{Helland1982}). Let $\mathcal F_{\ell}=\sigma(\xi_{n,1},\ldots,\xi_{n,\ell})$, we have to verify
\begin{align}
&\E\left(W^{\top}_{\ell-1}\xi^C_{n,\ell}\,\Big| \,\mathcal{F}_{\ell-1}\right)=0, \quad \ell\geq 1, \, n\geq 1,
\label{eq:mart1}\\
&  \sum_{\ell=1}^n\E\left\{\left(W^{\top}_{\ell-1}\xi^C_{n,\ell}\right)^2\,\Big| \,\mathcal{F}_{\ell-1}\right\}\tp 1,
\label{eq:mart2}\\
& L_n(\varepsilon)=\sum_{\ell=1}^n\E\left\{\left(W^{\top}_{\ell-1}\xi^C_{n,\ell}\right)^2 \mathbb{I}\left(\left\lvert  W^{\top}_{\ell-1}\xi^C_{n,\ell}\right\rvert  >\varepsilon\right)\,\Big| \,\mathcal{F}_{\ell-1}\right\}\tp 0, \quad \forall \varepsilon>0.
\label{eq:mart3}
\end{align}

First, \eqref{eq:mart1} is readily satisfied since $\xi_{n,1},\ldots,\xi_{n,n}$ are independent and centered.

Now, let $S_n=\sum_{\ell=1}^n Y_{n,\ell}$, where $Y_{n,\ell}
=\mathbb{E}\left\{
\left(W_{\ell-1}^{\top}\xi^C_{n,\ell}\right)^2
\,\Big| \,\mathcal{F}_{\ell-1}
\right\}$,
$\ell \in \{1,\ldots,n\}$. Since $\xi_{n,\ell}$ is independent of $\xi_{n,1},\ldots,\xi_{n,\ell-1}$ and is centered, we have
$ 
Y_{n,\ell}
=
\mathbb{E}\left\{
\left(W_{\ell-1}^{\top}\xi^C_{n,\ell}\right)^2
\,\Big| \,\mathcal{F}_{\ell-1}
\right\}  
=
W_{\ell-1}^{\top}
\mathbb{E}\left(\xi^C_{n,\ell}(\xi^C_{n,\ell})^{\top}\right)
W_{\ell-1}
=
W_{\ell-1}^{\top}\Xi_{\ell}^C W_{\ell-1}.
$ 
Therefore,
\begin{equation}
S_n
=
\sum_{\ell=1}^n W_{\ell-1}^{\top}\Xi_{\ell}^C W_{\ell-1} 
=
\frac{4}{n^2\sigma_n^2}
\sum_{\ell=1}^n
\left(
\sum_{s<\ell}\xi^C_{n,s}
\right)^{\top}
\Xi_{\ell}^C
\left(
\sum_{s<\ell}\xi^C_{n,s}
\right) 
=
\frac{4}{n^2\sigma_n^2}
\sum_{\ell=1}^n
\left\|  \sum_{s=1}^{\ell-1}
(\Xi_{\ell}^C)^{1/2}\xi^C_{n,s}
\right\| _2^2.
\label{eq:Sn-expression}
\end{equation}
Taking expectations, and noting that the vector
$(\Xi_{\ell}^C)^{1/2}\xi^C_{n,s}$ is centered and has covariance matrix $(\Xi_{\ell}^C)^{1/2}\Xi_s^C(\Xi_{\ell}^C)^{1/2}$, we get that
\begin{align*}
\mathbb{E}(S_n)
&=
\frac{4}{n^2\sigma_n^2}
\sum_{\ell=1}^n
\sum_{s<\ell}
\operatorname{tr}
\left\{
(\Xi_{\ell}^C)^{1/2}\Xi_s^C(\Xi_{\ell}^C)^{1/2}
\right\} \notag \\
&=
\frac{4}{n^2\sigma_n^2}
\sum_{\ell=1}^n
\sum_{s<\ell}
\operatorname{tr}\left(\Xi_{\ell}^C\Xi_s^C\right) =
\frac{2}{n^2\sigma_n^2}
\left[
\operatorname{tr}
\left\{
\left(\sum_{t=1}^n \Xi_t^C\right)^2
\right\}
-
\sum_{t=1}^n
\operatorname{tr}
\left\{
(\Xi_t^C)^2
\right\}
\right],
\end{align*}
where we have used the cyclic property of the trace. Recalling expression \eqref{eq:munsigmanxi} for $\sigma_n^2$, we obtain
\begin{equation}
\mathbb{E}(S_n)
=
1-
\frac{2}{n^2\sigma_n^2}
\sum_{t=1}^n
\operatorname{tr}
\left\{
(\Xi_t^C)^2
\right\}.
\label{eq:ESn}
\end{equation}
Now, note that for two symmetric matrices $A,B \geq 0$, from Von Neumann's Trace Theorem (see, e.g.,~Theorem 7.4.1.1 in \citeauthor{Horn_Johnson_1985}, \citeyear{Horn_Johnson_1985}), it follows that
\begin{equation}
\tr(AB)\ \le\ \| A\| \,\tr(B)
= \lambda_{1}(A)\,\tr(B).
\label{eq:ltrace-op}
\end{equation}
Then, since each $\Xi_t^C$ is symmetric and positive semidefinite, from \eqref{eq:ltrace-op} it follows that
$
\operatorname{tr}\left\{(\Xi_t^C)^2\right\}
\leq
\| \Xi_t^C\| \operatorname{tr}(\Xi_t^C)
=
\lambda_1(\Xi_t^C)\operatorname{tr}(\Xi_t^C)
\leq
\lambda_1^2(\Xi_t^C)d
\leq Md,
$
where the last inequality follows from \eqref{eq:eigencontrolXi}. Hence, $0\leq
2/(n^2\sigma_n^2)
\sum_{t=1}^n
\operatorname{tr}
\left\{
(\Xi_t^C)^2
\right\}
\leq
M/n
\to 0$, and thus, from \eqref{eq:ESn}, it follows that $\mathbb{E}(S_n)\to 1$. With this, to prove \eqref{eq:mart2}, we show that $\V(S_n)\to 0$, and applying Markov's inequality we get the result. Recall expression \eqref{eq:Sn-expression} for $S_n$. For each fixed $\ell$, the random vectors $(\Xi_{\ell}^C)^{1/2}\xi^C_{n,s}$, $s \in \{1,\ldots,n\}$, are independent and centered, with covariance matrices $(\Xi_{\ell}^C)^{1/2}\Xi_s^C(\Xi_{\ell}^C)^{1/2}$, so that \eqref{eq:vars2} gives the variance of the $\ell$-th summand in \eqref{eq:Sn-expression}. Therefore, we can bound 
\begin{equation}
\V(S_n)
\leq
\frac{16}{n^2\sigma_n^4}
\max_{1\leq \ell\leq n}
\biggl[4 \sum_{1\leq r<s\leq n}
\operatorname{tr}
\left\{(\Xi_{\ell}^C)^{1/2}
\Xi_s^C
\Xi_{\ell}^C
\Xi_r^C
(\Xi_{\ell}^C)^{1/2}\right\}+
\sum_{s=1}^n
\V\left(\left\| 
(\Xi_{\ell}^C)^{1/2}\xi^C_{n,s}
\right\| _2^2
\right)\biggr].
\label{eq:var-Sn-bound}
\end{equation}
From the cyclic property of the trace, \eqref{eq:eigencontrolXi} and \eqref{eq:ltrace-op}, we have that
\begin{equation}
\operatorname{tr}
\left\{
(\Xi_{\ell}^C)^{1/2}
\Xi_s^C
\Xi_{\ell}^C
\Xi_r^C
(\Xi_{\ell}^C)^{1/2}
\right\}
=
\operatorname{tr}
\left\{
\Xi_{\ell}^C
\Xi_s^C
\Xi_{\ell}^C
\Xi_r^C
\right\} \leq
\| \Xi_{\ell}^C\| ^2
\| \Xi_r^C\| 
\operatorname{tr}(\Xi_s^C)
\leq Md.
\label{eq:trace-four-bound}
\end{equation}
On the other hand, from \eqref{eq:normcompat}, \eqref{eq:eigencontrolXi} and \eqref{eq:orderxi4}, it follows that
\begin{equation}
\V
\left(
\left\| 
(\Xi_{\ell}^C)^{1/2}\xi^C_{n,s}
\right\| _2^2
\right)
\leq
\mathbb{E}
\left(
\left\| 
(\Xi_{\ell}^C)^{1/2}\xi^C_{n,s}
\right\| _2^4
\right) 
\leq
\| \Xi_{\ell}^C\| ^2
\mathbb{E}\left(\| \xi^C_{n,s}\| _2^4\right)
\leq Md^2.
\label{eq:variance-quadratic-bound}
\end{equation}
Introducing \eqref{eq:trace-four-bound} and \eqref{eq:variance-quadratic-bound} into
\eqref{eq:var-Sn-bound}, and recalling \eqref{eq:ordersigma}, we obtain  that $\V(S_n)
\leq
M\left(1/d+1/n\right)
\to0$, and since $\mathbb{E}(S_n)\to1$, we conclude that
\eqref{eq:mart2} holds.

Finally, we verify \eqref{eq:mart3}. 
From the fact that $x^2\mathbb{I}\{| x| >\varepsilon\}\leq x^4/\varepsilon^2$ for $x>0$, it follows that
$\E\{L_n(\varepsilon)\}\leq \varepsilon^{-2}T_n$, where
\begin{equation}
T_n=
\sum_{\ell=1}^n\E\{(W_{\ell-1}^{\top}\xi^C_{n,\ell})^4\}
=
\frac{16}{n^4\sigma_n^4}
\sum_{\ell=2}^n
\E\left\{
\left(
\sum_{j=1}^{\ell-1}
(\xi^C_{n,j})^{\top}\xi^C_{n,\ell}
\right)^4
\right\}.
\label{eq:Tn}
\end{equation}
Thus, it suffices to prove $T_n\to0$.
Let $a_j=(\xi^C_{n,j})^{\top}\xi^C_{n,\ell}$, for fixed $\ell\in\{1,\ldots, n\}$ and $j\in\{1,\ldots,\ell-1\}$.
Since the variables $\{a_{lj}\}_{j<\ell}$ are centered,
$\mathbb{E}
\left\{
\left(\sum_{j=1}^{\ell-1}a_j\right)^4
\right\}
=
\sum_{j=1}^{\ell-1}\mathbb{E}(a_j^4)
+
6\sum_{1\leq j<r\leq \ell-1}
\mathbb{E}(a_j^2a_r^2).
$ 
Consequently, from \eqref{eq:Tn} it follows that
\begin{equation}
\mathbb{E}(T_n)
=
\frac{16}{n^4\sigma_n^4}
\sum_{\ell=2}^n
\left\{
\sum_{j=1}^{\ell-1}\mathbb{E}(a_j^4)
+
6\sum_{1\leq j<r\leq \ell-1}
\mathbb{E}(a_j^2a_r^2)
\right\}.
\label{eq:ETn}
\end{equation}
Now, using independence, for $j\neq r$ we have that
\begin{equation}
\mathbb{E}(a_j^2a_r^2)
=
\mathbb{E}
\left\{
\mathbb{E}
\left(
a_j^2
\,\Big| \,
\xi^C_{n,\ell}
\right)
\mathbb{E}
\left(
a_r^2
\,\Big| \,
\xi^C_{n,\ell}
\right)
\right\}.
\label{eq:ajaj}
\end{equation}
For $s\in\{1,\ldots,\ell-1\}$, 
$\mathbb{E}
\left(
a_s^2
\,\Big| \,
\xi^C_{n,\ell}
\right)
=
\mathbb{E}
\left[
\left\{
(\xi^C_{n,s})^{\top}\xi^C_{n,\ell}
\right\}^2
\,\Big| \,
\xi^C_{n,\ell}
\right] =
(\xi^C_{n,\ell})^{\top}
\Xi_s^C
\xi^C_{n,\ell}
=
\left\| 
(\Xi_s^C)^{1/2}\xi^C_{n,\ell}
\right\| _2^2$.
Substituting this into \eqref{eq:ajaj}, and recalling \eqref{eq:normcompat}, \eqref{eq:eigencontrolXi} and \eqref{eq:orderxi4}
\begin{equation}
\mathbb{E}(a_j^2a_r^2)
=
\mathbb{E}
\left(
\left\| 
(\Xi_j^C)^{1/2}\xi^C_{n,\ell}
\right\| _2^2
\left\| 
(\Xi_r^C)^{1/2}\xi^C_{n,\ell}
\right\| _2^2
\right) \leq
\| \Xi_j^C\| 
\| \Xi_r^C\| 
\mathbb{E}\left(\| \xi^C_{n,\ell}\| _2^4\right)
\leq Md^2.
\label{eq:ajj-bound}
\end{equation}
For the fourth moments, from the Cauchy--Schwarz inequality, independence and \eqref{eq:orderxi4}, we can write
\begin{equation}
\mathbb{E}(a_j^4)=
\mathbb{E}
\left[
\left\{
(\xi^C_{n,j})^{\top}\xi^C_{n,\ell}
\right\}^4
\right] \leq
\mathbb{E}
\left(
\| \xi^C_{n,j}\| _2^4
\| \xi^C_{n,\ell}\| _2^4
\right)
=
\mathbb{E}\left(\| \xi^C_{n,j}\| _2^4\right)
\mathbb{E}\left(\| \xi^C_{n,\ell}\| _2^4\right)
\leq Md^4.
\label{eq:aj-four-bound}
\end{equation}
Introducing \eqref{eq:ajj-bound} and \eqref{eq:aj-four-bound} into \eqref{eq:ETn}, and recalling \eqref{eq:ordersigma}, we obtain
\begin{equation*}
\mathbb{E}(T_n)\leq
\frac{M}{n^4d^2}
\sum_{\ell=2}^n
\left\{
(\ell-1)d^4
+
3(\ell-1)(\ell-2)d^2
\right\}  \leq
\frac{M}{n^4d^2}
\left(
n^2d^4+n^3d^2
\right)
=
M\left(
\frac{d^2}{n^2}
+
\frac{1}{n}
\right)
\to0,
\end{equation*}
which proves \eqref{eq:mart3} and completes the proof.
\end{Proofx}

Next  we give  a result on the convergence rate of the terms involved in  Lemma~\ref{Lemma:varconsist}.
\begin{lemma}\label{lemma:widehatsigma-sigma}
Let $\Sigma_i=\left(\sigma_{i, ab}\right)$ and $\widehat{\Sigma}_i=\Bigl(\widehat{\sigma}_{i, ab}\Bigr)$, $a,b \in \{1,\ldots,d\}$, where $\Sigma_i$ and $\widehat{\Sigma}_i$ are defined in \eqref{eq:definitionsigmai} and \eqref{eq:WidehatSigma}, respectively. Suppose that  \eqref{eq:finitinessh4} and the conditions in  Lemma~\ref{Lemma:varconsist} hold, and that $f$ has continuous second-order partial derivatives in a neighborhood of $\eta_i$. Then $\left\lvert  \widehat{\sigma}_{i,ab}-\sigma_{i,ab}\right\rvert  =O_{\Prob}\left(n^{-1/2}\right)$, $a,b \in \{1,\ldots,d\}$, $ i \in \{1,\ldots,k\}$.
\end{lemma}
\begin{Proof}
First, for simplicity we drop the subscript $i$ so that $\Sigma$, $n$, $\theta$ and $\eta$ stand for a generic $\Sigma_i$, $n_i$, $\theta_i$ and $\eta_i$, respectively. Then, we will show that $\widehat{\sigma}_{ab}=\sigma_{ab}+O_{\Prob}\left(n^{-1/2}\right)$.
In order to prove this lemma, we make use of the calculations in the proof of  Lemma~\ref{Lemma:varconsist}. In that proof we have shown that $\widehat{\sigma}_{ab}$
can be written as a sum of terms of types \eqref{eq:type1}--\eqref{eq:type5}. First, note that since $f^{(a)}$ and $f^{(b)}$ have second derivatives in a neighborhood of $\eta$, it follows that $\left\| \nabla\left\{ f^{(r)}(\tau_j) - f^{(r)}(\eta) \right\}\right\| _2 \leq M \| \tau_j-\eta\| _2$, $r=a,b$.
Using the same arguments as in the proof of  Lemma~\ref{Lemma:varconsist}, and recalling inequalities \eqref{eq:inequality1} and \eqref{eq:maxtau-eta}, we obtain that the terms in \eqref{eq:type2} and \eqref{eq:type3} are $O_{\Prob}(n^{-1/2})$, and those in \eqref{eq:type4} and \eqref{eq:type5} are $O_{\Prob}\left(n^{-1}\right)$.

To complete the proof we must see that the term in \eqref{eq:type1} equals $\sigma_{ab}+O_{\Prob}(n^{-1/2})$. To show this, from \eqref{eq:type1sums} and \eqref{eq:YltSLNN}, we can express \eqref{eq:type1} as
\begin{equation*}
\begin{aligned}
\sum_{\ell=1}^q \sum_{t=1}^q \left\{(n-1)  \binom{n-1}{m_{\ell}}^{-1} \binom{n-1}{m_t}^{-1}  \binom{n}{m_{\ell}} \binom{m_{\ell}}{1} \binom{n-m_{\ell}}{m_t-1} U^{(\ell,t)}_1 \, f^{(a)}_{\ell}(\eta) f^{(b)}_t(\eta) \right\} +O_{\Prob}\left(n^{-1}\right),
\end{aligned}
\end{equation*}
where, as shown before, the $U$-statistic $U^{(\ell,t)}_1$ has symmetric kernel $K^{(\ell,t)}_1$, defined in \eqref{eq:Kc},
with $\E\left(K^{(\ell,t)}_1\right)=\E\left\{\widetilde{h}_1^{(\ell)}(X_1)\widetilde{h}_1^{(t)}(X_1)\right\}$. 
With this, recalling \eqref{eq:coefficientmlmt}, it suffices to check that $U_1^{(\ell,t)}$ satisfies $U$-statistics CLT \eqref{eq:assymptheta}, which holds if $\V\left(K_1^{(\ell,t)}\right)<\infty$.
To do so, note that if $X_1, \ldots, X_n$ are real-valued random variables with finite variances, and  
$S = \sum_{i=1}^n X_i$, the variance of $S$ satisfies the upper bound
$\V(S) \le \left\{\sum_{i=1}^n \sqrt{\V(X_i)}\right\}^{\!2}$.
Applying this inequality to $\V\left(K_1^{(\ell,t)}\right)$, it remains to show that
\begin{equation*}
\V\left[\left\{\widetilde{h}^{(\ell)}(X_{\alpha_1},X_{\beta_1},\ldots,X_{\beta_{m_{\ell}-1}}) \widetilde{h}^{(t)}(X_{\alpha_1},X_{\gamma_1},\ldots,X_{\gamma_{m_t-1}})\right\}\right]<\infty,
\end{equation*} for all fixed indices $\{\ell,t,\alpha_1,\beta_1,\ldots,\beta_{m_{\ell}-1},\gamma_1,\ldots,\gamma_{m_t-1}\}$. This follows from \eqref{eq:finitinessh4} since
\begin{align*}
\V&\left[\left\{\widetilde{h}^{(\ell)}(X_{\alpha_1},X_{\beta_1},\ldots,X_{\beta_{m_{\ell}-1}}) \,\widetilde{h}^{(t)}(X_{\alpha_1},X_{\gamma_1},\ldots,X_{\gamma_{m_t-1}})\right\}\right]\\
&\leq\E^{1/2}\left\{\left\lvert  \widetilde{h}^{(\ell)}(X_{\alpha_1},X_{\beta_1},\ldots,X_{\beta_{m_{\ell}-1}})\right\rvert  ^4\right\} \E^{1/2}\left\{\left\lvert  \widetilde{h}^{(t)}(X_{\alpha_1},X_{\gamma_1},\ldots,X_{\gamma_{m_t-1}})\right\rvert  ^4\right\}.
\end{align*}

\end{Proof}

\begin{Proofpx}{prop:RatioConsistentHD}
First, note that 
$(Q_n-\widehat{\mu}_n)/\widehat{\sigma}_n=({\sigma_n}/{\widehat{\sigma}_n}) \left\{({Q_n-\mu_n})/{\sigma_n}+ ({\mu_n-\widehat{\mu}_n})/{\sigma_n}\right\}$.
Therefore, from  Theorem~\ref{theorem:assymptoticdinfty} and Slutsky's Theorem, it suffices to prove that
\begin{align}
&(\widehat{\mu}_n-\mu_n)/\sigma_n \tp 0
\label{eq:ratiomun-mu},\\
&     \widehat{\sigma}^2_n/\sigma^2_n \tp 1.
\label{eq:ratiosigman}
\end{align}  
First, we prove \eqref{eq:ratiomun-mu}. Since $\widehat{\Sigma}-\Sigma_n$ is block-diagonal and $C=P_k\otimes I_d$, from  Lemma~\ref{lemma:widehatsigma-sigma}, \eqref{eq:compsamp}, \eqref{eq:ordersigma} and the definition of trace, it follows that
\begin{equation*}
(\widehat{\mu}_n-\mu_n)/{\sigma_n}={\tr\{ (\widehat{\Sigma}-\Sigma_n) C\}}/ {\sigma_n} = (1-1/k){\tr (\widehat{\Sigma}-\Sigma_n)} /{\sigma_n}= O_{\Prob}(\sqrt{d/n})=o_{\Prob}(1),
\end{equation*}
and then \eqref{eq:ratiomun-mu} holds.
Finally, we prove \eqref{eq:ratiosigman}. Let $\Delta_n=\widehat{\Sigma}-\Sigma_n$. From  Lemma~\ref{lemma:widehatsigma-sigma} and \eqref{eq:compsamp}, we have that
\begin{equation}
\delta_{n,ij}:=(\Delta_n)_{ij}=O_{\Prob}(n^{-1/2}).
\label{eq:orderdelta}
\end{equation}
With this, recalling that $C^2=C$ and the cyclic property of the trace
we obtain
\begin{equation*}
{\widehat{\sigma}^2_n}/{\sigma^2_n}={2\tr\{(\widehat{\Sigma} C)^2\}}/{\sigma^2_n}  = 1 + {4 \tr(C\Sigma_n C\, C\Delta_n C)}/{\sigma_n^2} + {2\tr\{(C\Delta_n C)^2\}}/{\sigma_n^2}.
\end{equation*}
By noting that $\tr\left(\Delta^2_n\right)=\sum_{i=1}^{kd} \sum_{j=1}^{kd} \delta_{n,ij} \delta_{n,ji}$, from \eqref{eq:orderdelta} we get that $\tr(\Delta_n^2)=O_{\Prob}(d^2/n)$. On the one hand, from the Cauchy--Schwarz inequality for the trace, $|\tr(AB)|\leq\{\tr(A^2)\tr(B^2)\}^{1/2}$ for symmetric $A$ and $B$, recalling that $\tr\{(C\Sigma_n C)^2\}=\sigma_n^2/2$ by \eqref{eq:sigman}, that $\tr\{(C\Delta_n C)^2\}\leq \tr(\Delta_n^2)$, since $C$ is an orthogonal projection, and \eqref{eq:ordersigma}, we have that
\begin{equation*}
{4 \lvert\tr(C\Sigma_n C\, C\Delta_n C)\rvert}/{\sigma_n^2} \leq 4\left[\tr\{(C\Sigma_n C)^2\}\tr\{(C\Delta_n C)^2\}\right]^{1/2}/\sigma_n^2 \leq 2\sqrt{2}\,\{\tr(\Delta_n^2)\}^{1/2}/\sigma_n=O_{\Prob}\left(\sqrt{d/n}\right).
\end{equation*}
On the other hand, from \eqref{eq:ordersigma}, we get that ${2\tr\{(C\Delta_n C)^2\}}/{\sigma_n^2}\leq {2\tr(\Delta_n^2)}/{\sigma_n^2}=O_{\Prob}(d/n)$. 
Therefore, since $d/n \to 0$,
we obtain \eqref{eq:ratiosigman} and the proof is complete.
\end{Proofpx}

\begin{remark}
The asymptotic results established in  Theorem~\ref{theorem:assymptoticdinfty} and  Proposition~\ref{prop:RatioConsistentHD} are derived under the regime $d\to\infty$ and
$d/n\to0$. These assumptions are used at several points of their proofs:

\noindent\textit{Asymptotic Normality (Theorem~\ref{theorem:assymptoticdinfty})}: The $d/n \to 0$ regime is required at three critical stages of the proof. First, in Step 2, to ensure that the higher-order error terms ($R$) from the Hoeffding decomposition and Taylor expansion are asymptotically negligible (see \eqref{eq:hr2to0} and \eqref{eq:hlhrto0}). Second, in Step 4, to guarantee that the diagonal term $Q_{n,1}$ vanishes in probability, since $\E(Q_{n,1}^2) = O(d/n)$. 
Third, after writing $Q_{n,2}$ as a martingale array, $d\to\infty$ is used to verify the conditional variance condition \eqref{eq:mart2}, while 
$d/n\to0$ is required to verify the conditional Lindeberg condition \eqref{eq:mart3}.

\noindent\textit{Covariance Estimation (Proposition~\ref{prop:RatioConsistentHD}):}  The estimation of squared trace terms introduces an accumulation of noise over $d$ dimensions, bounded by $\tr(\Delta_n^2)/\sigma_n^2 = O_{\Prob}(d/n)$. 
If $d/n \not\to 0$, this noise does not vanish, artificially inflating the denominator of the feasible ATS-ID.
\end{remark}

\section*{Acknowledgements}

The authors thank two anonymous referees for their constructive comments and suggestions which helped to improve the presentation. This research has been financed by research project PID2022-137818OB-I00 (Ministerio de Ciencia, Innovación y Universidades, Spain). The authors thank IMUS-Mar\'ia de Maeztu grant CEX2024-001517-M - Apoyo a Unidades de Excelencia María de Maeztu for supporting this research, funded by MICIU/AEI/10.13039/501100011033.

\bibliographystyle{cas-model2-names}
\bibliography{bibliografia}

\end{document}